\documentclass[%
  reprint,
  superscriptaddress,
]{revtex4-2}

\usepackage{amsmath}
\usepackage{amssymb}
\usepackage{bm,bbm}

\usepackage{graphicx}

\usepackage{physics}

\usepackage[bookmarks=false]{hyperref}
\hypersetup{colorlinks=true,citecolor=blue,linkcolor=red,%
urlcolor=blue,pdfstartview=FitH,bookmarksopen=true}

\usepackage[T1]{fontenc}
\usepackage[osf,sc]{mathpazo}
\usepackage{dsfont}

\usepackage{amsthm}
\newtheoremstyle{note}
{3pt}
{3pt}
{}
{}
{\itshape}
{:}
{.5em}
{}

\newtheorem{theorem}{Theorem}

\newtheorem{observation}[theorem]{Observation}

\DeclareMathOperator{\diag}{diag}

\newcommand{\me}{\mathrm{e}}
\newcommand{\mi}{\mathrm{i}}

\newcommand{\I}{\mathds{1}}

\renewcommand{\Bar}[1]{\overline{#1}}

\newcommand{\findover}[1][]{\underset{#1}{\mathrm{find}}}
\newcommand{\maxover}[1][]{\underset{#1}{\mathrm{max}}}

\newcommand{\minover}[1][]{\underset{#1}{\mathrm{min}}}
\newcommand{\subto}{\mathrm{~s.t.}}

\newcommand{\Sep}{\mathrm{SEP}}
\newcommand{\PPT}{\mathrm{PPT}}

\newcommand{\vl}{\bm{\lambda}}
\newcommand{\vp}{\bm{p}}
\newcommand{\vw}{\bm{w}}

\newcommand{\cH}{\mathcal{H}}

\newcommand{\cQ}{\mathcal{Q}}
\newcommand{\cR}{\mathcal{R}}
\newcommand{\cS}{\mathcal{S}}

\newcommand{\dC}{\mathds{C}}

\newcommand{\dR}{\mathds{R}}

\newcommand{\KK}{\mathrm{KK}}

\begin{document}

\title{Characterizing high-dimensional quantum contextuality}

\author{Xiao-Dong Yu}
\email{yuxiaodong@sdu.edu.cn}
\affiliation{Department of Physics, Shandong University, Jinan 250100, China}
\author{Isadora Veeren}
\email{veeren@cbpf.br}
\affiliation{Naturwissenschaftlich-Technische Fakult\"at, Universit\"at Siegen,
Walter-Flex-Str. 3, D-57068 Siegen, Germany}
\affiliation{Centro Brasileiro de Pesquisas F\'\i sicas (CBPF),
Rua Doutor Xavier Sigaud 150, 22290-180 Rio de Janeiro, Brazil}
\author{Otfried G\"uhne}
\email{otfried.guehne@uni-siegen.de}
\affiliation{Naturwissenschaftlich-Technische Fakult\"at, Universit\"at Siegen,
Walter-Flex-Str. 3, D-57068 Siegen, Germany}

\date{\today}

\begin{abstract}
  As a phenomenon encompassing measurement incompatibility and Bell 
  nonlocality, quantum contextuality is not only central to our understanding 
  of quantum mechanics, but also an essential resource in many quantum 
  information processing tasks. The dimension-dependent feature of quantum 
  contextuality is known ever since its discovery, but there is still a lack of 
  systematic methods for characterizing this fundamental feature. In this work, 
  we propose a systematic and reliable method for certifying the 
  high-dimensional advantages of quantum contextuality. In theory, our work 
  gives a complete characterization of the dimension-constrained quantum 
  contextual behavior, and particularly its nonconvex structure is revealed.  
  In application, our method can be used for dimensionality certification of 
  quantum information processing systems, and also for concentrating the 
  quantum contextual behavior into lower-dimensional systems.
\end{abstract}

\maketitle

\section{Introduction}
%
Qubits are the basic building blocks in many quantum information processing 
protocols. However, treating a real quantum system as a qubit is not only 
unnecessary, but also merely an approximation in practice. In recent years, 
experimental progress enabled the control of high-dimensional quantum systems 
and theoretical works demonstrated potential advantages of information 
processing in the high-dimensional case 
\cite{Erhard.etal2020,Cozzolino.etal2019}.  Consequently, many efficient 
methods have been developed to certify the high-dimensional advantages of 
various quantum resources, such as quantum entanglement 
\cite{Bavaresco.etal2018}, quantum coherence \cite{Ringbauer.etal2018}, and 
Bell nonlocality \cite{Navascues.Vertesi2015}.

As a phenomenon encompassing measurement incompatibility and Bell nonlocality, 
quantum contextuality is not only central to our understanding of quantum 
mechanics \cite{Kochen.Specker1967,Peres1993,Budroni.etal2022}, but also an 
essential resource in many quantum information processing tasks, such as in 
quantum computation \cite{Howard.etal2014,Bravyi.etal2018,Bravyi.etal2020}, in 
quantum cryptography \cite{BechmannPasquinucci.Peres2000,Svozil2009}, and in 
random number generation \cite{Abbott.etal2012,Kulikov.etal2017,Um.etal2013}.  
The study of quantum contextuality originates from the work of Kochen and 
Specker, which is now referred to as the Kochen-Specker theorem 
\cite{Kochen.Specker1967}.
The modern theory-independent framework for quantum contextuality was proposed 
by Klyachko \textit{et al.} \cite{Klyachko.etal2008} and further developed to 
the state-independent scenario by Cabello \textit{et al.} 
\cite{Cabello2008,Badziag.etal2009,Yu.Oh2012}. In these works, noncontextuality 
inequalities are discovered. These inequalities are obeyed by noncontextual 
hidden variable (NCHV) models, but can be violated by quantum mechanics.  This 
experimentally testable framework greatly promotes both the theoretical and 
experimental study of quantum contextuality.

Ever since the discovery of quantum contextuality, people have noticed 
its dimension-dependent feature. For example, both Kochen and Specker 
\cite{Kochen.Specker1967} and Bell \cite{Bell1966} proved that quantum 
contextuality does not exist in two-dimensional quantum systems. Also, 
every proof of the Kochen-Specker theorem, or more generally, every 
state-independent proof of quantum contextuality, is dimension dependent 
\cite{Ramanathan.Horodecki2014,Cabello.etal2015}. For some noncontextuality
inequalities, it was discussed in detail how the largest quantum violation
depends on the dimensionality of the quantum system \cite{Guehne.etal2014}.
Recently, Ray \textit{et al.} investigated the problem of calculating
finite-dimensional lower bounds of a family of noncontextuality 
inequalities \cite{Ray.etal2021}.

Despite all these efforts, there is still a lack of systematic methods for 
certifying the high-dimensional advantages of quantum contextuality. On one 
hand, systematic methods to calculate the $d$-dimensional violation of general 
noncontextuality inequalities are still missing. On the other hand, it is not 
known whether using linear inequalities gives a complete characterization of 
dimension-constrained quantum contextual behavior.

In this work, we solve both of these problems. We first propose an efficient 
and reliable method for certifying whether or not a quantum contextual behavior 
can result from a $d$-dimensional quantum system. This provides a complete 
characterization of dimension-constrained quantum contextual behaviors.  In  
particular, we prove that not all dimension-constrained quantum contextual 
behaviors can be characterized by the linear inequality method, which reveals 
a significant difference between quantum contextual behaviors with and without 
dimension constraints. Then, we show that the proposed method can also be 
adapted for calculating $d$-dimensional violation of general noncontextuality 
inequalities.  Finally, we discuss the implications of our results for 
certifying the dimensionality of quantum information processing systems and 
concentrating the quantum contextuality into lower-dimensional systems.

\section{Preliminaries}
%
In quantum contextuality theory, a measurement context 
$\{s_1,s_2,\dots,s_\alpha\}$ is a set of compatible measurements, which are 
jointly measurable.  An event $e=(o_1,o_2,\dots,o_\alpha\mid 
s_1,s_2,\dots,s_\alpha)$ means in a joint measurement 
$\{s_1,s_2,\dots,s_\alpha\}$ the outcome of $s_i$ is $o_i$ for 
$i=1,2,\dots,\alpha$. Two events $(o_1,o_2,\dots,o_\alpha\mid 
s_1,s_2,\dots,s_\alpha)$ and $(o'_1,o'_2,\dots,o'_\beta\mid 
s'_1,s'_2,\dots,s'_\beta)$ are called exclusive if there exist $a,b$ such that 
$s_a=s'_b$ but $o_a\ne o'_b$. The events and their exclusive relation can be 
depicted by a graph, which is called the exclusivity graph.  Mathematically, 
a graph $G$ is denoted by $(V,E)$, where $V$ is the set of vertices and $E$ is 
the set of edges, i.e., unordered pairs $\{i,j\}$ for some $i,j\in V$ and $i\ne 
j$.  In the exclusivity graph, the vertices represent the set of events 
$\{e_1,e_2,\dots,e_n\}$ and edges connect pairs of exclusive events; see 
Fig.~\ref{fig:KK} for an example graph $G_{\KK}$ and its connection to quantum 
contextuality. Notably, $G_{\KK}$ can be viewed as a variant of the ``bug'' 
graph from the original Kochen-Specker argument \cite{Kochen.Specker1967} and 
has been used for revealing quantum contextuality of almost all qutrit states 
\cite{Kurzynski.Kaszlikowski2012}.

In Ref.~\cite{Cabello.etal2014}, an important connection between the graph 
theory and quantum contextuality was discovered. Consider the sum of 
probabilities $\sum_{i=1}^np(e_i)$, where $p(e_i)$ are the probabilities of the 
corresponding events. In an NCHV model, determinism and exclusivity imply that 
this sum is upper bounded by
\begin{equation}
  \begin{aligned}
    \alpha(G)~:=~&\maxover[b_i]  && \sum_{i=1}^n b_i,\\
      &\subto && b_i=0 \text{~or~} 1,\\
      & && b_ib_j=0 \text{~~for~~} \{i,j\}\in E.
    \end{aligned}
    \label{eq:Independence}
\end{equation}
Here, $\alpha(G)$ is the so-called independence number of graph $G$, which 
corresponds to the maximum number of mutually unconnected vertices in $G$.
In quantum theory, the events $e_i$ are represented by projectors $P_i$ and two 
events are exclusive if they are orthogonal, i.e., $P_iP_j=0$. Without loss of 
generality, we can always assume that the state is pure and $P_i$ are rank one 
for studying noncontextuality inequalities \cite{Budroni.etal2022}. Thus, the 
quantum bound of $\sum_{i=1}^np(e_i)=\sum_{i=1}^n\Tr(\rho P_i)$ is given by
\begin{equation}
  \begin{aligned}
    \vartheta(G)~:=~&\maxover[\ket{\varphi},\ket{\psi_i}]
    && \sum_{i=1}^n \abs{\braket{\varphi}{\psi_i}}^2,\\
      &\subto && \braket{\psi_i}{\psi_j}=0 \text{~~for~~} \{i,j\}\in E.
    \end{aligned}
  \label{eq:Lovasz}
\end{equation}
$\vartheta(G)$ equals to the so-called Lov\'asz number of graph $G$ 
\cite{Lovasz1979}.  Here and in the following, we assume that $\ket{\varphi}$ 
and $\ket{\psi_i}$ are always normalized unless otherwise stated. We call the 
vectors $(\ket{\psi_1},\ket{\psi_2},\dots,\ket{\psi_n})$ satisfying 
the constraints in Eq.~\eqref{eq:Lovasz} a (rank-one) projective representation 
of $G$
\footnote{In graph theory, $(\ket{\psi_1},\ket{\psi_2},\dots,\ket{\psi_n})$ is 
called an orthonormal representation of the complement graph $\Bar{G}$.}.
One can easily see from the definition that $\alpha(G)\le\vartheta(G)$. If 
$\alpha(G)$ is strictly smaller than $\vartheta(G)$, it implies that there is 
a gap between the classical (NCHV) bound and the maximally achievable quantum 
value.  Consequently, noncontextuality inequalities can be constructed 
\cite{Yu.Tong2014,Cabello2016}. More generally, the set of all realizable 
probabilities in quantum theory
\begin{equation}
  \begin{aligned}
    \cQ(G)=\{(p_1,p_2,\dots,p_n)\mid p_i=\abs{\braket{\varphi}{\psi_i}}^2,\\
      \braket{\psi_i}{\psi_j}=0
    \text{~~for~~} \{i,j\}\in E\}
  \end{aligned}
  \label{eq:thetaBody}
\end{equation}
corresponds to the so-called theta body of graph $G$ 
\cite{Groetschel.etal1986}. If $(p_1,p_2,\dots,p_n)\in\cQ(G)$, it is then 
called a quantum contextual behavior or simply a quantum behavior.

\begin{figure}
  \centering
  \includegraphics[width=0.35\textwidth]{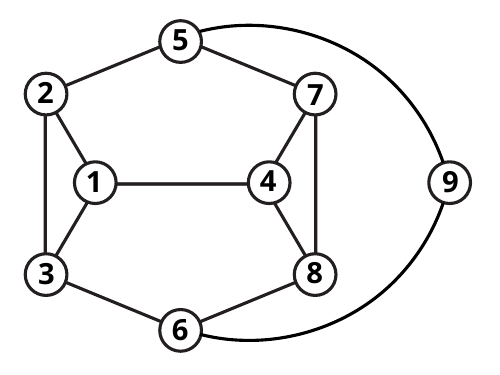}
  \caption{The nine-vertex graph $G_{\KK}$. In quantum theory, every vertex 
    represents a projector $P_i$, and two projectors $P_i$ and $P_j$ are 
    connected by an edge when $P_iP_j=0$, which means that they are exclusive 
    events.  The independence number of $G_{\KK}$ is  $\alpha(G_{\KK})=3$, 
    which can be achieved when $b_1,b_6,b_7$ in Eq.~\eqref{eq:Independence} 
    take the value one.  The Lov\'asz number of $G_{\KK}$ is 
    $\vartheta(G_{\KK})=4.4704$, which can be achieved in a $4$-dimensional 
  quantum system.}
  \label{fig:KK}
\end{figure}

The quantum bound $\vartheta(G)$ and the quantum behaviors $\cQ(G)$ reveal 
a characteristic feature of quantum mechanics, namely, quantum contextuality.  However, 
this feature depends on the dimension of the quantum system. For example, 
quantum contextuality does not exist in two-dimensional systems, and the 
maximization in Eq.~\eqref{eq:Lovasz} also depends on the dimension of the 
quantum system. In this work, we give a systematic study on the 
dimension-dependent nature of quantum contextuality, and complete methods for 
characterizing the set of $d$-dimensional quantum behaviors
\begin{equation}
  \begin{aligned}
    \cQ_d(G)=\{(p_1,p_2,\dots,p_n)\mid\ket{\varphi},\ket{\psi_i}\in\dC^d,\\
      p_i=\abs{\braket{\varphi}{\psi_i}}^2,
      ~\braket{\psi_i}{\psi_j}=0
    \text{~~for~~} \{i,j\}\in E\}
  \end{aligned}
  \label{eq:finiteThetaBody}
\end{equation}
are proposed.

\section{Finite-dimensional quantum contextuality}
%
We start with developing a method for verifying whether a behavior 
$\vp=(p_1,p_2,\dots,p_n)$ can result from a $d$-dimensional quantum system.  
Remarkably, there are two different kinds of verification. Verifying that 
$\vp\notin\cQ_d(G)$ ($\vp\in\cQ_d(G)$) is an outer (inner) approximation 
problem, which means an affirmative conclusion would imply that 
$\vp\notin\cQ_d(G)$ ($\vp\in\cQ_d(G)$), otherwise the verification is 
inconclusive. In this work we will consider both cases.  The starting point is 
the following observation.
\begin{observation}
  The probabilities $(p_1,p_2,\dots,p_n)$ are a $d$-dimensional quantum 
  behavior, i.e., $(p_1,p_2,\dots,p_n)\in\cQ_d(G)$, if and only if there exists 
  a Hermitian matrix $[X_{ij}]_{i,j=0}^{n}$ satisfying that
  \begin{subequations}
    \begin{align}
      \label{eq:verify1}
      & X_{0i}=X_{i0}=\sqrt{p_i}\quad\text{~for~} i=1,2,\dots,n,\\
      \label{eq:verify2}
      & X_{ii}=1\quad\text{~for~} i=0,1,2,\dots,n,\\
      \label{eq:verify3}
      & X_{ij}=0\quad\text{~for~}\{i,j\}\in E,\\
      \label{eq:verify4}
      & X\ge 0,~\rank(X)\le d.
    \end{align}
  \end{subequations}
  \label{thm:verify}
\end{observation}

The main idea for proving Observation~\ref{thm:verify} is that the dimension of 
the vector space generated by $\{\ket{v_i}\}_{i=0}^n$ equals to the rank of the 
corresponding Gram matrix $\left[\braket{v_i}{v_j}\right]_{i,j=0}^n$; for the 
detailed proof, please see Appendix~\ref{app:proof}.

With Observation~\ref{thm:verify}, we can construct a convex program which is 
necessary and sufficient for $(p_1,p_2,\dots,p_n)$ being a $d$-dimensional 
quantum behavior. Suppose that $X\in\dC^{(n+1)\times (n+1)}$ satisfies the 
constraints~(\ref{eq:verify1}--\ref{eq:verify4}). As $X\ge 0$ and $\rank(X)\le 
d$, one can construct a purification of $X$ with a $d$-dimensional auxiliary 
system, i.e., there exists an unnormalized pure state 
$\ket{\chi}\in\dC^{n+1}\otimes\dC^d$, such that
\begin{equation}
  \Tr_2(\ketbra{\chi})=X,
  \label{eq:purification}
\end{equation}
where $\Tr_2(\cdot)$ is the partial trace operation on the second subsystem 
$\dC^d$. Following the ideas in Refs.~\cite{Yu.etal2021,Yu.etal2022}, we 
consider the two-copy extension
\begin{equation}
  \Phi_{AB}=\ketbra{\chi}_A\otimes\ketbra{\chi}_B,
  \label{eq:twocopy}
\end{equation}
then $\Phi_{AB}$ is an unnormalized state in $\cH_A\otimes\cH_B$ with 
$\cH_A=\cH_B=\dC^{n+1}\otimes\dC^d$. Moreover, $\Phi_{AB}$ is in the symmetric 
subspace:
\begin{equation}
  V_{AB}\Phi_{AB}=\Phi_{AB},
\end{equation}
where the swap operator $V_{AB}$ is defined to satisfy that 
$V_{AB}\ket{\chi}_A\ket{\xi}_B=\ket{\xi}_A\ket{\chi}_B$ for any pair of states 
$\ket{\chi}$ and $\ket{\xi}$. By imposing the other constraints in 
Eqs.~(\ref{eq:verify1}--\ref{eq:verify3}), one can easily see that if 
$(p_1,p_2,\dots,p_n)\in\cQ_d(G)$, the following convex program is feasible:
\begin{equation}
  \begin{aligned}
    &\findover  &&\Phi_{AB}\in\Sep\\
    &\subto && V_{AB}\Phi_{AB}=\Phi_{AB},~\Tr(\Phi_{AB})=(n+1)^2,\\
    & &&
    \Tr_A\big[(\ketbra{i}{j}\otimes\I_d\otimes\I_B)\Phi_{AB}\big]
    =\frac{\mu_{ij}}{n+1}\Tr_A\big[\Phi_{AB}\big],
  \end{aligned}
  \label{eq:verifyConvex}
\end{equation}
where $\mu_{ij}$ denote all the known elements of $X_{ij}$, i.e.,
$\mu_{0i}=\mu_{i0}=\sqrt{p_i}$ for $i=1,2,\dots,n$, $\mu_{ii}=1$ for 
$i=0,1,2,\dots,n$, and $\mu_{ij}=0$ for $\{i,j\}\in E$, and $\Sep$ denotes the 
set of unnormalized separable states. Moreover, this convex program is also 
sufficient for $(p_1,p_2,\dots,p_n)\in\cQ_d(G)$ \cite{Yu.etal2022}. From 
Eq.~\eqref{eq:verifyConvex}, a complete hierarchy of semidefinite programs 
(SDPs) can be constructed for outer approximating $\cQ_d(G)$, and the lowest 
order is replacing $\Phi_{AB}\in\Sep$ with $\Phi_{AB}\in\PPT$. If any of these 
SDPs is infeasible, it will imply that $(p_1,p_2,\dots,p_n)\notin\cQ_d(G)$.

\begin{figure}
  \centering
  \includegraphics[width=.3\textwidth]{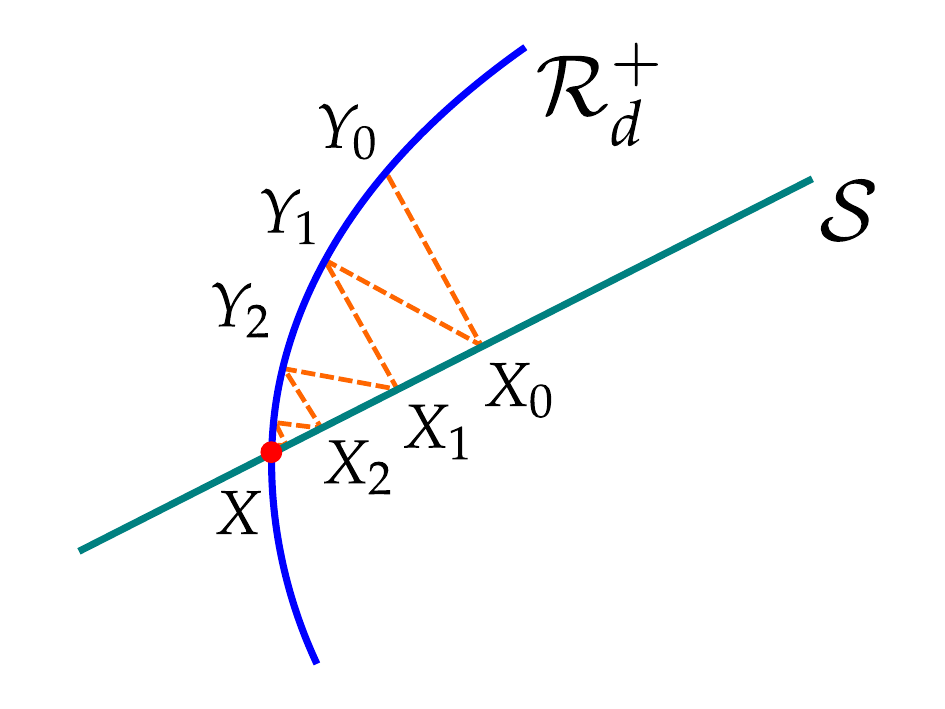}
  \caption{Illustration of the inner approximation method. The problem is 
    equivalent to finding $X\in\cS\cap\cR_d^+$, for which we minimize the 
    distance between points in $\cS$ and $\cR_d^+$. We first randomly choose 
    a point $Y_0\in\cR_d^+$ and find $X_0\in\cS$ that minimizes the distance 
    between $Y_0$ and $\cS$, i.e., $X_0=\arg\min_{X\in\cS}\norm{X-Y_0}_F$.  
    Similarly, with $X_0$ we can then find  
    $Y_1=\arg\min_{Y\in\cR_d^+}\norm{X_0-Y}_F$.  Repeating the above procedure, 
    i.e., $X_i=\arg\min_{X\in\cS}\norm{X-Y_i}_F$ and 
    $Y_{i+1}=\arg\min_{Y\in\cR_d^+}\norm{X_i-Y}_F$, we get a converging 
    sequence $\norm{X_i-Y_i}_F$.  If the limit is zero, we obtain the desired 
    $X\in\cS\cap\cR_d^+$ and Observation~\ref{thm:verify} implies that the 
    corresponding $(p_1,p_2,\dots,p_n)\in\cQ_d(G)$.
  }
  \label{fig:seesaw}
\end{figure}

For the inner approximation, i.e., verifying $(p_1,p_2,\dots,p_n)\in\cQ_d(G)$, 
we note that Observation~\ref{thm:verify} can be viewed as a semidefinite 
variant of the so-called low-rank matrix recovery, which is a rapidly developing 
field in computer science \cite{Davenport.Romberg2016}.  For our problem, the 
matrix size is relatively small and thus efficient methods can be constructed.  
The main idea of our method is illustrated in Fig.~\ref{fig:seesaw}. Let $\cS$ 
and $\cR_d^+$ denote the set of Hermitian matrices satisfying 
Eqs.~(\ref{eq:verify1}--\ref{eq:verify3}) and satisfying 
Eq.~\eqref{eq:verify4}, respectively,
then there exists $(p_1,p_2,\dots,p_n)\in\cQ_d(G)$ if and only if the solution 
of the following optimization problem is zero:
\begin{equation}
  \begin{aligned}
    &\minover[X,Y] \quad && \norm{X-Y}_F\\
    &\subto && X\in\cS,~Y\in\cR_d^+,
  \end{aligned}
  \label{eq:seesaw}
\end{equation}
where $\norm{\cdot}_F$ denotes the Frobenius norm. From Eq.~\eqref{eq:seesaw} 
an alternating optimization algorithm for verifying 
$(p_1,p_2,\dots,p_n)\in\cQ_d(G)$ can be constructed. More technical details of 
the inner and outer approximation algorithms can be found in 
Appendices~\ref{app:innerVerify} and \ref{app:outerVerify}.


To illustrate the power of our methods, we consider the nine-vertex graph 
$G_{\KK}$ in Fig.~\ref{fig:KK} and the behaviors
\begin{subequations}
  \begin{align}
    \vp_1=&\left(\frac{1}{3},\frac{1}{3},\frac{1}{3},
      \frac{1}{3},\frac{1}{3},\frac{1}{3},
    \frac{1}{3},\frac{1}{3},\frac{1}{3}\right),\\
    \vp_2=&\left(\frac{1}{2},\frac{1}{4},\frac{1}{4},
      \frac{1}{2},0,0,
    \frac{1}{4},\frac{1}{4},1\right),\\
    \vp_3=&\left(\frac{5}{12},\frac{7}{24},\frac{7}{24},
      \frac{5}{12},\frac{1}{6},\frac{1}{6},
    \frac{7}{24},\frac{7}{24},\frac{2}{3}\right).
    \label{eq:behaviors}
  \end{align}
\end{subequations}
One can prove that $\vp_1,\vp_2\in\cQ_d(G_{\KK})$ with the inner approximation 
methods, and $\vp_3\notin\cQ_d(G_{\KK})$ with the outer approximation method.
Note also that $\vp_3$ is a mixture of $\vp_1$ and $\vp_2$, i.e., 
$\vp_3=(\vp_1+\vp_2)/2$. This reveals a remarkable difference between the 
quantum behaviors with and without dimension constraints: the set of general 
quantum behaviors $\cQ(G)$ is convex but the $d$-dimensional counterpart 
$\cQ_d(G)$ may not be.
Similar nonconvexity results also exist for other quantum resources; see 
Refs.~\cite{Bowles.etal2014,Donohue.Wolfe2015,Sikora.etal2016,Mao.etal2022} for 
some examples and their applications. Another remarkable property is that 
although the quantum bound $\vartheta(G_\KK)$ can already be achieved when 
$d=4$, there exist quantum behaviors not in $\cQ_4(G_{\KK})$. One such example 
is
\begin{equation}
  \vp_4=\left(\frac{1}{3},\frac{1}{3},\frac{1}{3},
    0,\frac{2}{3},\frac{1}{3},
  0,0,\frac{1}{3}\right).
  \label{eq:5behavior}
\end{equation}

\section{Inequality method}
%
The standard method for characterizing quantum contextuality without dimension 
constraints relies on noncontextuality inequalities. These are linear 
inequalities closely related to the so-called weighted Lov\'asz number 
\cite{Cabello.etal2014}. Similarly, one can also characterize the 
$d$-dimensional quantum contextuality with the following quantity, which can be 
viewed as the $d$-dimensional weighted Lov\'asz number:
\begin{equation}
  \begin{aligned}
    \vartheta_d(G, \vw)~
    :=~&\maxover[\vp] && \vw\cdot\vp\\
    &\subto && \vp\in\cQ_d(G),
    \end{aligned}
  \label{eq:LovaszFinite}
\end{equation}
where the weights $\vw=(w_1,w_2,\dots,w_n)\in\dR^n$ and 
$\vw\cdot\vp=\sum_{i=1}^nw_ip_i$.  Unlike $\cQ(G)$, this inequality method is 
no longer sufficient for characterizing $\cQ_d(G)$ because of the nonconvexity 
property proved above.  Furthermore, contrary to $\cQ(G)$ 
\cite{Cabello.etal2015}, there are no longer reasons to assume that all $w_i$ 
are nonnegative for characterizing $\cQ_d(G)$.  What 
Eq.~\eqref{eq:LovaszFinite} characterizes is actually the convex hull of 
$\cQ_d(G)$ and thus the inequality method is less general than the direct 
method given above.  Nevertheless, an inequality is sometimes more suitable for 
experimental tests. In the following, we show that our method can also be 
adapted for calculating the bound in Eq.~\eqref{eq:LovaszFinite}.

By taking advantage of Observation~\ref{thm:verify}, we get the following 
equivalent form of $\vartheta_d(G,\vw)$:
\begin{equation}
  \begin{aligned}
    &\maxover[X] \qquad &&\sum_{i=1}^nw_i\abs{X_{0i}}^2\\
    &\subto && X_{ii}=1\quad\text{~for~} i=0,1,2,\dots,n,\\
    & && X_{ij}=0\quad\text{~for~}\{i,j\}\in E,\\
    & && X\ge 0,~\rank(X)\le d.
  \end{aligned}
  \label{eq:LovaszFiniteQuad}
\end{equation}
Similarly, this rank-constrained optimization can be transformed to the convex 
optimization \cite{Yu.etal2022}
\begin{equation}
  \begin{aligned}
    &\maxover[\Phi_{AB}] &&\Tr[W_{AB}\Phi_{AB}]\\
    &\subto && 
    \Phi_{AB}\in\Sep,\,V_{AB}\Phi_{AB}=\Phi_{AB},\,\Tr(\Phi_{AB})=(n+1)^2,\\
    & && \Tr_A\big[(\ketbra{i}{j}\otimes\I_d\otimes\I_B)\Phi_{AB}\big]=0
    \text{~~for~} \{i,j\}\in E,\\
    & && \Tr_A\big[(\ketbra{i}\otimes\I_d\otimes\I_B)\Phi_{AB}\big]
    =\frac{1}{n+1}\Tr_A\big[\Phi_{AB}\big]\\
    & && \quad\text{~for~} i=0,1,\dots,n,
  \end{aligned}
  \label{eq:LovaszFiniteConvex}
\end{equation}
where
$W_{AB} =\frac{1}{2}(\sum_{i=1}^nw_i\ketbra{0}{i} 
\otimes\I_d\otimes\ketbra{i}{0}\otimes\I_d+\text{H.c.})$, and $\text{H.c.}$ 
denotes the Hermitian conjugate of the previous term.

From Eq.~\eqref{eq:LovaszFiniteConvex} a complete SDP hierarchy can be 
constructed for upper bounding $\vartheta_d(G,\vw)$, but the low-order 
relaxations may not give good enough bounds. Thus, we provide another method, 
which is not always complete but may result in better bounds when low-order 
relaxations are considered. Consider the Gram matrix of 
$(\ket{\varphi},c_1\ket{\psi_1},c_2\ket{\psi_2},\dots,c_n\ket{\psi_n})$, where 
$c_i=\braket{\psi_i}{\varphi}$.  Similarly to Observation~\ref{thm:verify}, one 
can prove that $\vartheta_d(G,\vw)$ is upper bounded by the optimization
\begin{equation}
  \begin{aligned}
    &\maxover[X] \quad &&\sum_{i=1}^n w_iX_{ii}\\
    &\subto && X_{ii}=X_{0i}=X_{i0}\quad\text{~for~} i=1,2,\dots,n,\\
    & && X_{ij}=0\quad\text{~for~}\{i,j\}\in E,\\
    & && X_{00}=1,~X\ge 0,~\rank(X)\le d,
  \end{aligned}
  \label{eq:LovaszFiniteRayRelax}
\end{equation}
where the constraints $X_{ii}=X_{0i}=X_{i0}$ result from the conditions that 
$\bra{\varphi}c_i\ket{\psi_i}=\bra{\psi_i}c_i^*c_i\ket{\psi_i} 
=\abs{\braket{\psi_i}{\varphi}}^2$.  Let $\tilde\vartheta_d(G,\vw)$ denote the 
solution of Eq.~\eqref{eq:LovaszFiniteRayRelax}, then one can easily see that 
$\vartheta_d(G,\vw)\le\tilde\vartheta_d(G,\vw)$.  In Ref.~\cite{Ray.etal2021}, 
it was claimed that $\tilde\vartheta_d(G,\vw)=\vartheta_d(G,\vw)$ when 
$\cQ_d(G)$ is not empty. This is, however, not true. An explicit counterexample 
is shown below.

From Eqs.~(\ref{eq:LovaszFiniteQuad},\,\ref{eq:LovaszFiniteRayRelax}),
efficient inner and outer approximation methods can be similarly constructed
for lower and upper bounding $\vartheta_d(G,\vw)$ and $\tilde\vartheta_d(G,\vw)$.
As an example, we still consider graph $G_\KK$ in Fig.~\ref{fig:KK} and the 
case that $\vw=(1,1,\dots,1)$, for which $\vartheta_d(G,\vw)$ and 
$\tilde\vartheta_d(G,\vw)$ are denoted by $\vartheta_d(G)$ and 
$\tilde\vartheta_d(G)$, respectively. One can prove that 
$\vartheta_3(G_{\KK})=3.3333$ by showing that $3.3333$ is also both a lower 
bound and an upper bound (up to numerical precision).
In addition, graph $G_{\KK}$ also provides an explicit example of 
$\tilde\vartheta_d(G)\ne\vartheta_d(G)$.
This can be proved by constructing a matrix $X$ satisfying all the
constraints in Eq.~\eqref{eq:LovaszFiniteRayRelax} and
$\sum_{i=1}^nX_{ii}=3.3380$, which then implies that
$\tilde\vartheta_3(G_\KK)\ge 3.3380>\vartheta_3(G_\KK)=3.3333$.
See Appendices~\ref{app:lowerBound} and \ref{app:upperBound} for the technical 
details of the lower and upper bounding algorithms.

\section{Discussion and conclusion}
%
Given the extensive theoretical and experimental studies on the 
high-dimensional advantages of quantum resources in recent years 
\cite{Ali-Khan.etal2007,Neeley.etal2009,Lanyon.etal2009,Dada.etal2011,
Kues.etal2017,Wang.etal2018,Ecker.etal2019,Kong.etal2023}, our method can be 
used in various ways. First, our result provide a new approach for constructing 
so-called dimension witnesses
\cite{Brunner.etal2008,Gallego.etal2010,Ahrens.etal2012,Ahrens.etal2012,
Hendrych.etal2012,Brunner.etal2013,Bowles.etal2014,DAmbrosio.etal2014}. These 
are linear or nonlinear inequalities, which can be used to certify the 
experimenter's coherent control on a certain amount of 
levels in quantum information processing. 
With our outer approximation method, a violation of the 
inequality $\sum_iw_ip_i\ge\vartheta_d(G,\vw)$ can be rigorously proved, which  
would in turn certify that the amount of controllable levels is larger than $d$.  
Second, we consider the so-called contextuality contraction, which aims to 
achieve the same degree of contextuality with a lower-dimensional system and 
thus make the utilization of quantum contextuality more experimentally 
accessible \cite{Liu.etal2023}. One can easily see that our inner approximation 
method can be directly used for reducing the dimension and our outer 
approximation method can be used for calculating the limit of contextuality 
contraction. Finally, on a more abstract level, our results may elucidate the 
role of the quantum dimension in information processing. The original 
definition of the Lov\'asz number was motivated by notions of communication 
capacity and for contextuality  connections to communication tasks have been 
established \cite{Spekkens.etal2009,Cubitt.etal2010}.  Combining these concepts 
may lead to a deeper understanding  of high-dimensional quantum information 
processing, as well as novel applications.

In conclusion, we have provided powerful methods to characterize quantum 
contextual behavior under dimension constraints. Our method gives a complete 
characterization of dimension-constrained quantum contextual behaviors, and 
particularly we show that not all quantum contextual behaviors can be 
characterized by the linear inequality method.
As applications, our method can be used for dimensionality certification of 
quantum information processing systems, and also for concentrating the quantum 
contextuality behavior into lower-dimensional systems.

We thank Kishor Bharti, Ad\'an Cabello, and Zhen-Peng Xu for discussions.
This work has been supported by the Deutsche Forschungsgemeinschaft
(DFG, German Research Foundation, project numbers 447948357 
and 440958198), 
the Sino-German Center for Research Promotion (Project M-0294), 
the ERC (Consolidator Grant 683107/TempoQ),
and the German Ministry of Education and Research
(Project QuKuK, BMBF Grant No. 16KIS1618K).
X.D.Y. acknowledges support by the National Natural Science Foundation of China
(Grants No. 12205170 and No. 12174224)
and the Shandong Provincial Natural Science Foundation of China (Grant No. ZR2022QA084).

\onecolumngrid
\bigskip

\appendix

\newtheorem{manualobservationinner}{Observation}
\newenvironment{manualobservation}[1]{%
  \renewcommand\themanualobservationinner{#1}%
\manualobservationinner}{\endmanualobservationinner}

\section{Proof of Observation~\ref{thm:verify}}
\label{app:proof}
%
\begin{manualobservation}{\ref{thm:verify}}
  The probabilities $(p_1,p_2,\dots,p_n)$ are a $d$-dimensional quantum 
  behavior, i.e., $(p_1,p_2,\dots,p_n)\in\cQ_d(G)$, if and only if there exists 
  a Hermitian matrix $[X_{ij}]_{i,j=0}^{n}$ satisfying that
  \begin{subequations}
    \begin{align}
      \label{eq:verify1A}
      & X_{0i}=X_{i0}=\sqrt{p_i}\quad\text{~for~} i=1,2,\dots,n,\\
      \label{eq:verify2A}
      & X_{ii}=1\quad\text{~for~} i=0,1,2,\dots,n,\\
      \label{eq:verify3A}
      & X_{ij}=0\quad\text{~for~}\{i,j\}\in E,\\
      \label{eq:verify4A}
      & X\ge 0,~\rank(X)\le d.
    \end{align}
  \end{subequations}
  \label{thm:verifyA}
\end{manualobservation}

We first prove the ``only if'' part. Recall that the set of $d$-dimensional 
quantum behaviors $\cQ_d(G)$ is defined as
\begin{equation}
  \cQ_d(G)=\{(p_1,p_2,\dots,p_n)\mid\ket{\varphi},\ket{\psi_i}\in\dC^d,~
    p_i=\abs{\braket{\varphi}{\psi_i}}^2,
    ~\braket{\psi_i}{\psi_j}=0
  \text{~~for~~} \{i,j\}\in E\}.
  \label{eq:finiteThetaBodyA}
\end{equation}
Suppose that $(p_1,p_2,\dots,p_n)\in\cQ_d(G)$, then the definition in 
Eq.~\eqref{eq:finiteThetaBodyA} implies that there exist 
$\ket{\varphi},\ket{\psi_i}\in\dC^d$ such that $\braket{\psi_i}{\psi_j}=0$ for 
$\{i,j\}\in E$ and $p_i=\abs{\braket{\varphi}{\psi_i}}^2$. Let
\begin{equation}
  \ket{\tilde\psi_i}=\me^{-\mi\gamma_i}\ket{\psi_i},
\end{equation}
where $\gamma_i=\arg{\braket{\varphi}{\psi_i}}$, then we have that 
$\braket{\varphi}{\tilde\psi_i}=\sqrt{p_i}$. For simplicity, we let 
$\ket{\tilde\psi_0}=\ket{\varphi}$, and consider the Gram matrix of 
$(\ket{\tilde\psi_0},\ket{\tilde\psi_1},\ket{\tilde\psi_2},
\dots,\ket{\tilde\psi_n})$ defined by
\begin{equation}
  X_{ij}=\braket{\tilde\psi_i}{\tilde\psi_j}
  \quad\text{~for~}i,j=0,1,2,\dots,n.
  \label{eq:Gram}
\end{equation}
Then one can easily verify that $X$ satisfies 
conditions~(\ref{eq:verify1A}--\ref{eq:verify3A}). Furthermore, the matrix $X$ 
can also be written as
\begin{equation}
  X=\sum_{i,j=0}^nX_{ij}\ketbra{i}{j}
  =\sum_{i,j=0}^n\ket{i}\braket{\tilde\psi_i}{\tilde\psi_j}\bra{j}
  =L^\dagger L.
\end{equation}
where $L=\sum_{i=0}^n\ketbra{\tilde\psi_i}{i}$ is a $d\times (n+1)$ matrix, 
because $\ket{\tilde\psi_i}\in\dC^d$.  Thus, we have that $X\ge 0$ and 
$\rank(X)\le\rank(L)\le d$, corresponding to condition~\eqref{eq:verify4A}.

To prove the ``if'' part, suppose that $X$ satisfies all the 
conditions~(\ref{eq:verify1A}--\ref{eq:verify4A}), then it follows that $X$ 
admits a unitary decomposition
\begin{equation}
  X=U\diag(\lambda_0,\lambda_1,\dots,\lambda_{d-1},0,0,\dots,0)U^\dagger,
\end{equation}
where $\lambda_i\ge 0$. Let
\begin{equation}
  \ket{\psi_i}
  =\left(\sum_{k=0}^{d-1}\sqrt{\lambda_k}\ketbra{k}\right)U^\dagger\ket{i}
  =\sum_{k=0}^{d-1}\left(\sqrt{\lambda_k}\bra{k}U^\dagger\ket{i}\right)\ket{k}
\end{equation}
for $i=0,1,\dots,n$.  By letting $\ket{\varphi}=\ket{\psi_0}$, one can easily 
verify that $\ket{\varphi},\ket{\psi_i}$ satisfy the conditions in 
Eq.~\eqref{eq:finiteThetaBodyA}.

\section{Inner approximation for $d$-dimensional quantum contextuality}
\label{app:innerVerify}
%

For the inner approximation, according to Observation~\ref{thm:verify}, we only 
need to find a Hermitian matrix $[X_{ij}]_{i,j=0}^n$ such that
\begin{equation}
  X_{ij}=\mu_{ij},~X\ge 0,~\rank(X)\le d,
  \label{eq:verify}
\end{equation}
where $\mu_{ij}$ denote all the known elements of $X_{ij}$ in 
Eqs.~(\ref{eq:verify1}--\ref{eq:verify3}), i.e., $\mu_{0i}=\mu_{i0}=\sqrt{p_i}$ 
for $i=1,2,\dots,n$, $\mu_{ii}=1$ for $i=0,1,2,\dots,n$, and $\mu_{ij}=0$ for 
$\{i,j\}\in E$.

The first method is based on the fact that Eq.~\eqref{eq:verify} is feasible if 
and only if the solution of the following optimization problem is zero:
\begin{equation}
  \begin{aligned}
    &\minover[X,Y]  \quad && \norm{X-Y}_F\\
    &\subto && X=X^\dagger,~X_{ij}=\mu_{ij},\\
    &       && Y\ge 0,~\rank(Y)\le d,
  \end{aligned}
  \label{eq:innerN}
\end{equation}
where $\norm{\cdot}_F$ denotes the Frobenius norm (also known as 
Hilbert-Schmidt norm or Schatten-$2$ norm). Then, one can randomly choose 
a rank-$d$ semidefinite matrix $Y_\star$ as the initial value and alternatingly 
perform the following two optimization problems:
\begin{subequations}
  \begin{equation}
    \begin{aligned}
      &\minover[X]  \quad && \norm{X-Y_\star}_F\\
      &\subto && X=X^\dagger,~X_{ij}=\mu_{ij},
    \end{aligned}
    \label{eq:EckartYoungMirsky1}
  \end{equation}
  \begin{equation}
    \begin{aligned}
      &\minover[Y]  \quad && \norm{X_\star-Y}_F\\
      &\subto && Y\ge 0,~\rank(Y)\le d,
    \end{aligned}
    \label{eq:EckartYoungMirsky2}
  \end{equation}
\end{subequations}
where $X_\star$ in Eq.~\eqref{eq:EckartYoungMirsky2} is an optimal solution 
to Eq.~\eqref{eq:EckartYoungMirsky1} and $Y_\star$ in 
Eq.~\eqref{eq:EckartYoungMirsky1} is an optimal solution to 
Eq.~\eqref{eq:EckartYoungMirsky2}. One can easily see that the first 
minimization Eq.~\eqref{eq:EckartYoungMirsky1} is achieved when
\begin{equation}
  \begin{aligned}
      & X_{0i}=\sqrt{p_i}&&\text{~for~} i=1,2,\dots,n,\\
      & X_{ii}=1&&\text{~for~} i=0,1,2,\dots,n,\\
      & X_{ij}=0&&\text{~for~}\{i,j\}\in E,\\
      & X_{ij}=Y_{\star ij}&&\text{~for others}.
  \end{aligned}
\end{equation}
For the second minimization Eq.~\eqref{eq:EckartYoungMirsky2}, let 
$X_\star=\sum_{i=0}^n\lambda_i\ketbra{\phi_i}$ be the eigendecomposition of the 
Hermitian matrix $X_\star$ and the eigenvalues 
$(\lambda_0,\lambda_1,\dots,\lambda_n)$ are sorted in descending order, then 
the maximization is achieved when \cite[Appendix~F]{Xu.etal2021}
\begin{equation}
  Y=\sum_{i\in N_d} \lambda_i\ketbra{\phi_i},
\end{equation}
where $N_d=\{i \mid 0 \le i\le d-1,~\lambda_i\ge 0\}$, i.e., the summation is 
over the largest $d$ nonnegative eigenvalues, and if the number of nonnegative 
eigenvalues is smaller than $d$, the summation is over all the nonnegative 
eigenvalues. One can iteratively solve  
Eqs.~(\ref{eq:EckartYoungMirsky1},\,\ref{eq:EckartYoungMirsky2}) until 
convergence.  If the minimization convergences to zero, this means the final 
$X_\star$ satisfies all constraints in Eq.~\eqref{eq:verify}.

The second method is based on the fact that Eq.~\eqref{eq:verify} is feasible 
if and only if the solution of the following optimization problem is zero 
\cite[Sec~4.5]{Dattorro2005}:
\begin{equation}
  \begin{aligned}
    &\minover[X,P]  \quad && \Tr(XP)\\
    &\subto &&  X\ge 0,~X_{ij}=\mu_{ij},\\
    &       &&  P\ge 0,~P^2=P,~\rank(P)=n-d+1.
  \end{aligned}
  \label{eq:innerP}
\end{equation}
Then, one can randomly choose a rank-$(n-d+1)$ projector $P_\star$ as the 
initial value and alternatingly perform the following two optimization 
problems:
\begin{subequations}
  \begin{equation}
    \begin{aligned}
      &\minover[X]  \quad && \Tr(XP_\star)\\
      &\subto &&  X\ge 0,~X_{ij}=\mu_{ij},
    \end{aligned}
    \label{eq:Dattorro1}
  \end{equation}
  \begin{equation}
    \begin{aligned}
      &\minover[P]  \quad && \Tr(X_\star P)\\
      &\subto &&  P\ge 0,~P^2=P,~\rank(P)=n-d+1,
    \end{aligned}
    \label{eq:Dattorro2}
  \end{equation}
\end{subequations}
where $X_\star$ in Eq.~\eqref{eq:Dattorro2} is an optimal solution to 
Eq.~\eqref{eq:Dattorro1} and $P_\star$ in Eq.~\eqref{eq:Dattorro1} is an 
optimal solution to Eq.~\eqref{eq:Dattorro2}. One can easily see that the first 
optimization~\eqref{eq:Dattorro1} is an SDP, for which efficient algorithms 
exist \cite{Boyd.Vandenberghe2004}, and a solution to the second 
optimization~\eqref{eq:Dattorro2} is $P=\sum_{i=d}^n\ketbra{\phi_i}$, where 
$\ket{\phi_i}$ are from the eigendecomposition of 
$X_\star=\sum_{i=0}^n\lambda_i\ketbra{\phi_i}$ 
($\lambda_0\ge\lambda_1\ge\lambda_2\ge\dots\ge\lambda_n$). One can iteratively 
solve  Eqs.~(\ref{eq:Dattorro1},\,\ref{eq:Dattorro2}) until convergence.  If 
the minimization convergences to zero, this means the final $P_\star$ and 
$X_\star$ satisfy the constraints in Eq.~\eqref{eq:innerP} and $\Tr(X_\star 
P_\star)=0$ and thus $X_\star$ satisfies all constraints in 
Eq.~\eqref{eq:verify}.

\section{Outer approximation for $d$-dimensional quantum contextuality}
\label{app:outerVerify}
%

For the outer approximation, we prove in the main text that 
$(p_1,p_2,\dots,p_n)\in\cQ_d(G)$ if and only if there exists $\Phi_{AB}$ such 
that the following optimization problem is feasible:
\begin{equation}
  \begin{aligned}
    &\findover \quad &&\Phi_{AB}\in\Sep\\
    &\subto && V_{AB}\Phi_{AB}=\Phi_{AB},~\Tr(\Phi_{AB})=(n+1)^2,\\
    & && \Tr_A\big[(\ketbra{i}{j}\otimes\I_d\otimes\I_B)\Phi_{AB}\big]
    =\frac{\mu_{ij}}{n+1}\Tr_A\big[\Phi_{AB}\big],
  \end{aligned}
  \label{eq:verifyConvexApp}
\end{equation}
where $\mu_{ij}$ denote all the known elements of $X_{ij}$ in 
Eqs.~(\ref{eq:verify1}--\ref{eq:verify3}), i.e., $\mu_{0i}=\mu_{i0}=\sqrt{p_i}$ 
for $i=1,2,\dots,n$, $\mu_{ii}=1$ for $i=0,1,2,\dots,n$, and $\mu_{ij}=0$ for 
$\{i,j\}\in E$, and $\Sep$ denotes the set of all unnormalized separable 
states. From Eq.~\eqref{eq:verifyConvexApp}, a complete SDP hierarchy can be 
constructed for outer approximating $\cQ_d(G)$, and the lowest order is 
replacing $\Phi_{AB}\in\Sep$ with $\Phi_{AB}\in\PPT$ \cite{Yu.etal2022}. After 
applying the inherent symmetry in the problem, the PPT relaxation of the 
problem can be written as the following feasibility problem \cite{Yu.etal2022}:
\begin{equation}
  \begin{aligned}
    &\findover\quad  && \Phi_I,\Phi_V\\
    &\subto && \Phi_I+\Phi_V\ge 0,~\Phi_I-\Phi_V\ge 0,\\
    &       && \Phi_I^{T_{A_1}}\ge 0,
	       ~\Phi_I^{T_{A_1}}+d\Phi_V^{T_{A_1}}\ge 0,\\
    &       && \Phi_V=V_{A_1B_1}\Phi_I,~d^2\Tr(\Phi_I)+d\Tr(\Phi_V)=(n+1)^2,\\
    &       && \Tr_{A_1}\big[(\ketbra{i}{j}\otimes\I_{n+1})
	       (d\Phi_I+\Phi_V)\big]
	       =\frac{\mu_{ij}}{n+1}\Tr_{A_1}\big[d\Phi_I+\Phi_V\big],
  \end{aligned}
\end{equation}
where $\Phi_I$ and $\Phi_V$ are Hermitian operators on 
$\cH_{A_1}\otimes\cH_{B_1}=\dC^{n+1}\otimes\dC^{n+1}$, i.e., they are 
$(n+1)^2\times(n+1)^2$ Hermitian matrices. Actually, as all parameters in the 
SDP are real, we can take $\Phi_I$ and $\Phi_V$ as symmetric real matrices. To 
quantify the feasibility of the above problem, one can also rewrite it as the 
following maximization problem:
\begin{equation}
  \begin{aligned}
    &\maxover[\Phi_I,\Phi_V, \eta]\quad  && \eta\\
    &\subto && \Phi_I+\Phi_V\ge\eta\I,~\Phi_I-\Phi_V\ge\eta\I,\\
    &       && \Phi_I^{T_{A_1}}\ge\eta\I,
	       ~\Phi_I^{T_{A_1}}+d\Phi_V^{T_{A_1}}\ge\eta\I,\\
    &       && \Phi_V=V_{A_1B_1}\Phi_I,~d^2\Tr(\Phi_I)+d\Tr(\Phi_V)=(n+1)^2,\\
    &       && \Tr_{A_1}\big[(\ketbra{i}{j}\otimes\I_{n+1})
	       (d\Phi_I+\Phi_V)\big]
	       =\frac{\mu_{ij}}{n+1}\Tr_{A_1}\big[d\Phi_I+\Phi_V\big].
  \end{aligned}
\end{equation}
When the solution $\eta_{\max}<0$, we can assert that  
$(p_1,p_2,\dots,p_n)\notin\cQ_d(G)$.

\section{Lower bounds for $\vartheta_d(G,\vw)$ and $\tilde\vartheta_d(G,\vw)$}
\label{app:lowerBound}
%

To derive a lower bound for $\vartheta_d(G,\vw)$, we take advantage of 
Eq.~\eqref{eq:LovaszFiniteQuad} in the main text, which we rewrite below
\begin{equation}
  \begin{aligned}
    &\maxover[X] \qquad &&\sum_{i=1}^nw_i\abs{X_{0i}}^2\\
    &\subto && X_{ii}=1\quad\text{~for~} i=0,1,2,\dots,n,\\
    & && X_{ij}=0\quad\text{~for~}\{i,j\}\in E,\\
    & && X\ge 0,~\rank(X)\le d,
  \end{aligned}
  \label{eq:LovaszFiniteQuadApp}
\end{equation}
where $[X_{ij}]_{i,j=0}^n$ is an $(n+1)\times (n+1)$ Hermitian matrix.
One can easily see that if $\vartheta_d(G,\vw)$ exists, i.e., if $G$ admits 
a $d$-dimensional projective representation, the optimization on $X$ in 
Eq.~\eqref{eq:LovaszFiniteQuadApp} is equivalent to the following optimization 
on $X,Y$, when $\eta\to+\infty$:
\begin{equation}
  \begin{aligned}
    &\maxover[X,Y] \qquad &&\sum_{i=1}^nw_i\abs{X_{0i}}^2 -\eta\norm{X-Y}_F^2\\
    &\subto && X_{ii}=1\quad\text{~for~} i=0,1,2,\dots,n,\\
    & && X_{ij}=0\quad\text{~for~}\{i,j\}\in E,\\
    & && X=X^\dagger,\\
    & && Y\ge 0,~\rank(Y)\le d,
  \end{aligned}
  \label{eq:LovaszFiniteSeesaw}
\end{equation}
where $X$ and $Y$ are $(n+1)\times (n+1)$ Hermitian matrices and 
$\norm{\cdot}_F$ is the Frobenius norm. This is because that whenever $X\ne Y$, 
the negative penalty term $-\eta\norm{X-Y}_F^2$ in the objective function will 
dominate for large enough $\eta$. When $X=Y$, Eq.~\eqref{eq:LovaszFiniteSeesaw} 
reduces to Eq.~\eqref{eq:LovaszFiniteQuadApp}.

To solve Eq.~\eqref{eq:LovaszFiniteSeesaw}, we perform an alternating 
maximization over $X$ and $Y$. For fixed $Y=Y_\star$, we find an $X_\star$, 
which achieves the optimization
\begin{equation}
  \begin{aligned}
    &\maxover[X] \qquad &&\sum_{i=1}^nw_i\abs{X_{0i}}^2 
    -\eta\norm{X-Y_\star}_F^2\\
    &\subto && X_{ii}=1\quad\text{~for~} i=0,1,2,\dots,n,\\
    & && X_{ij}=0\quad\text{~for~}\{i,j\}\in E,\\
    & && X=X^\dagger.
  \end{aligned}
  \label{eq:QuadAlt1}
\end{equation}
This can be analytically solved because
\begin{equation}
  \sum_{i=1}^nw_i\abs{X_{0i}}^2-\eta\norm{X-Y_\star}_F^2
  =\sum_{i=1}^n\left(\frac{w_i}{2}\abs{X_{0i}}^2+\frac{w_i}{2}\abs{X_{i0}}^2\right)
  -\eta\sum_{i,j}\abs{X_{ij}-Y_{\star ij}}^2.
\end{equation}
When $\eta>\frac{w_i}{2}$ for all $i$, the maximization is achieved when
\begin{equation}
  \begin{aligned}
    & X_{ii}=1\quad\text{~for~} i=0,1,2,\dots,n,\\
    & X_{ij}=0\quad\text{~for~}\{i,j\}\in E,\\
    & X_{0i}=\frac{2\eta Y_{\star 0i}}{2\eta-w_i},~
    X_{i0}=\frac{2\eta Y_{\star i0}}{2\eta-w_i}\quad\text{~for~} 
    i=1,2,\dots,n,\\
    & X_{ij}=Y_{\star ij}\quad\text{~for~others}.
  \end{aligned}
\end{equation}
Similarly, for fixed $X_\star$, we find an $Y_\star$, which achieves the 
optimization
\begin{equation}
  \begin{aligned}
    &\maxover[Y] \qquad &&\sum_{i=1}^nw_i\abs{X_{\star 0i}}^2 
    -\eta\norm{X_\star -Y}_F^2\\
    &\subto && Y\ge 0,~\rank(Y)\le d.
  \end{aligned}
  \label{eq:QuadAlt2}
\end{equation}
Let $X_\star=\sum_{i=0}^n\lambda_i\ketbra{\phi_i}$ be the eigendecomposition of 
the Hermitian matrix $X_\star$ and the eigenvalues 
$(\lambda_0,\lambda_1,\dots,\lambda_n)$ are sorted in descending order, then 
the maximization is achieved when \cite[Appendix~F]{Xu.etal2021}
\begin{equation}
  Y=\sum_{i\in N_d} \lambda_i\ketbra{\phi_i},
\end{equation}
where $N_d=\{i \mid 0 \le i\le d-1,~\lambda_i\ge 0\}$, i.e., the summation is 
over the largest $d$ nonnegative eigenvalues, and if the number of nonnegative 
eigenvalues is smaller than $d$, the summation is over all the nonnegative 
eigenvalues.

In practice, it is usually good enough to choose some fixed $\eta$ that is 
neither too large nor too small. If $\eta$ is too small, the maximum term may 
not close to zero when the maximum is obtained. If $\eta$ is too large, the 
weight of $\sum_{i=1}^nw_i\abs{X_{\star 0i}}^2$ will be small and sensitive to 
numerical error. For example, one can choose $\eta=\max_i\abs{w_i}$ or 
$\eta=10\max_i\abs{w_i}$. Initially, one randomly chooses a rank-$d$ 
semidefinite matrix $Y_\star$.  Then, one can perform the alternating 
maximization over $M$ and $X$ in Eqs.~(\ref{eq:QuadAlt1},\,\ref{eq:QuadAlt2}) 
until convergence. At last, it is usually necessary to take $\eta\to+\infty$, 
i.e., run the minimization
\begin{equation}
  \begin{aligned}
    &\minover[X,Y] \qquad &&\norm{X-Y}_F^2\\
    &\subto && X_{ii}=1\quad\text{~for~} i=0,1,2,\dots,n,\\
    & && X_{ij}=0\quad\text{~for~}\{i,j\}\in E,\\
    & && X=X^\dagger,\\
    & && Y\ge 0,~\rank(Y)\le d,
  \end{aligned}
\end{equation}
further to find a better low-rank approximation $X$ (when the minimization is 
zero).  In our example, by taking $\eta=10$, we can easily obtain the lower 
bound $\vartheta_3(G_{\KK})\ge 3.3333$. In general, when $\eta$ is smaller the 
alternating algorithm convergences faster, and when $\eta$ is larger a higher 
precision lower bound is obtained.

For the most widely used case that all $w_i>0$, we provide an alternative 
method for calculating lower bounds of $\vartheta_d(G,\vw)$. Let $M$ be the 
Gram matrix of 
$(\sqrt{w_1}\ket{\psi_1},\sqrt{w_2}\ket{\psi_2},\dots,\sqrt{w_n}\ket{\psi_n})$, 
i.e.,
\begin{equation}
  M_{ij}=\sqrt{w_iw_j}\braket{\psi_i}{\psi_j}\quad\text{ for } i,j=1,2,\dots,n.
\end{equation}
From the definition of $\ket{\psi_i}$, we get that $M\ge 0$, $M_{ii}=w_i$ for 
all $i$ and $M_{ij}=0$ for all $\{i,j\}\in E$.  Further, the constraints 
$\ket{\psi_i}\in\dC^d$ imply that $\rank(M)\le d$.  At last, we have the 
condition
\begin{equation}
  \max_{\ket{\varphi}}\sum_{i=1}^nw_i\abs{\braket{\varphi}{\psi_i}}^2=
  \lambda_{\max}\left(\sum_{i=1}^nw_i\ketbra{\psi_i}\right)
  =\lambda_{\max}(M),
\end{equation}
where the last equality results from the relation that 
$\lambda_{\max}(LL^\dagger)=\lambda_{\max}(L^\dagger L)$ for 
$L=\sum_{i=1}^n\sqrt{w_i}\ketbra{\psi_i}{i}$ 
\cite[Theorem~1.3.22]{Horn.Johnson2012}.  Thus, we prove that 
$\vartheta_d(G,\vw)$ is upper bounded by the optimization
\begin{equation}
  \begin{aligned}
    &\maxover[M] \qquad &&\lambda_{\max}(M)\\
    &\subto && M_{ii}=w_i\quad\text{~for~} i=1,2,\dots,n,\\
    & && M_{ij}=0\quad\text{~for~}\{i,j\}\in E,\\
    & && M\ge 0,~\rank(M)\le d,
  \end{aligned}
  \label{eq:LovaszFiniteLambda}
\end{equation}
where $[M_{ij}]_{i,j=1}^n$ is an $n\times n$ Hermitian matrix and 
$\lambda_{\max}(M)$ denotes the largest eigenvalue of $M$. Actually, one can 
easily show that Eq.~\eqref{eq:LovaszFiniteLambda} gives the exact value of 
$\vartheta_d(G,\vw)$ with a similar argument as in 
Observation~\ref{thm:verify}.

To get a lower bound of $\vartheta_d(G,\vw)$ from 
Eq.~\eqref{eq:LovaszFiniteLambda}, one needs to accomplish two tasks: finding 
a low-rank matrix satisfying some SDP constraints and making the largest 
eigenvalue of such matrices as large as possible. To this end, we show that if 
$\vartheta_d(G,\vw)$ exists, i.e., if $G$ admits a $d$-dimensional projective 
representation, Eq.~\eqref{eq:LovaszFiniteLambda} is equivalent to the 
following maximization in the limit of $\varepsilon\to 0^+$:
\begin{equation}
  \begin{aligned}
    &\maxover[M,X] \qquad
    &&\varepsilon^{-1}[\Tr(MX)-\sum_{i=1}^nw_i]\\
    &\subto && M_{ii}=w_i\quad\text{for}~i=1,2,\dots,n,\\
    & && M_{ij}=0\quad\text{for}~\{i,j\}\in E,\\
    & && M\ge 0,~\vl(X)=\vp_{\varepsilon,d},
  \end{aligned}
  \label{eq:LovaszFiniteSeesawLambda}
\end{equation}
where $\vl(\cdot)$ denotes the spectrum and
\begin{equation}
  \vp_{\varepsilon,d}=(1+\varepsilon,
    \underbrace{1,\dots,1}_{d-1~\text{times}},
  \underbrace{0,\dots,0}_{n-d~\text{times}}).
  \label{eq:vp}
\end{equation}
To see the reason, let $M$ and $X$ be the variables that achieve the 
maximization.  Suppose that 
$\vl(M)=\vl_M=(\lambda_1,\lambda_2,\dots,\lambda_n)$ and 
$\lambda_1\ge\lambda_2\ge\dots\ge\lambda_n\ge 0$. Then, we have 
$\sum_{i=1}^n\lambda_i=\Tr(M)=\sum_{i=1}^nw_i$, and the maximum equals to 
$\varepsilon^{-1}(\vl_M\cdot\vl_{\varepsilon,d}-\sum_{i=1}^nw_i)
=\lambda_1+\varepsilon^{-1}(\sum_{i=1}^d\lambda_i-\sum_{i=1}^nw_i)
=\lambda_{\max}(M)-\varepsilon^{-1}(\sum_{i=d+1}^n\lambda_i)$. Thus, in the 
limit of $\varepsilon\to 0^+$, the maximization is achieved when $\lambda_i=0$ 
for $i=d+1,d+2,\dots,n$, i.e., $\rank(M)\le d$, and further the maximum is 
$\lambda_{\max}(M)$.

To solve Eq.~\eqref{eq:LovaszFiniteSeesawLambda}, we perform an alternating 
maximization over $M$ and $X$. For fixed $X=X_\star$, we find an $M_\star$, 
which achieves the optimization
\begin{equation}
  \begin{aligned}
    &\maxover[M] \qquad
    &&\varepsilon^{-1}[\Tr(MX_\star)-\sum_{i=1}^nw_i]\\
    &\subto && M_{ii}=w_i\quad\text{for}~i=1,2,\dots,n,\\
    & && M_{ij}=0\quad\text{for}~\{i,j\}\in E,\\
    & && M\ge 0.
  \end{aligned}
  \label{eq:LambdaAlt1}
\end{equation}
Similarly, for fixed $M_\star$, we find an $X_\star$, which achieves the 
optimization
\begin{equation}
  \begin{aligned}
    &\maxover[X] \qquad
    &&\varepsilon^{-1}[\Tr(M_\star X)-\sum_{i=1}^nw_i]\\
    &\subto && \vl(X)=\vp_{\varepsilon,d}.
  \end{aligned}
  \label{eq:LambdaAlt2}
\end{equation}
The optimization in Eq.~\eqref{eq:LambdaAlt1} is an SDP, for which efficient 
algorithms exist \cite{Boyd.Vandenberghe2004}. The optimization in 
Eq.~\eqref{eq:LambdaAlt2} can be solved explicitly with 
$X_\star=U_\star\diag(\vp_{\varepsilon,d})U_\star^\dagger$, where the unitary 
matrix $U_\star$ results from the eigendecomposition 
$M_\star=U_\star\diag(\vl_\star)U_\star^\dagger$ and the eigenvalues 
$\vl_\star$ are sorted in descending order.

In practice, it is usually good enough to choose some fixed $\varepsilon$ that 
is neither too large nor too small, for example, $\varepsilon=\min_i\abs{w_i}$ 
or $\varepsilon=\min_i\abs{w_i}/10$.  Initially, one randomly chooses a unitary 
matrix $U_\star$ and let 
$X_\star=U_\star\diag(\vp_{\varepsilon,d})U_\star^\dagger$.  Then, one can 
perform the alternating maximization over $M$ and $X$ in 
Eqs.~(\ref{eq:LambdaAlt1},\,\ref{eq:LambdaAlt2}) until convergence. At last, 
one needs to verify whether the obtained $M_\star$ is (approximately) of rank 
no larger than $d$. If so, we obtained a lower bound of $\vartheta_d(G,\vw)$, 
otherwise we repeat the whole procedure with a different initial unitary matrix 
$U_\star$.

Similarly, we can also calculate lower bounds of $\tilde\vartheta_d(G,\vw)$, 
which is defined as
\begin{equation}
  \begin{aligned}
    &\maxover[X] \qquad &&\sum_{i=1}^n w_iX_{ii}\\
    &\subto && X_{ii}=X_{0i}\quad\text{~for~} i=1,2,\dots,n,\\
    & && X_{ij}=0\quad\text{~for~}\{i,j\}\in E,\\
    & && X_{00}=1,~X\ge 0,~\rank(X)\le d.
  \end{aligned}
  \label{eq:LovaszFiniteRayRelaxApp}
\end{equation}
Similar to Eq.~\eqref{eq:LovaszFiniteSeesaw}, one can construct an alternating 
optimization based on (theoretically in the limit $\eta\to+\infty$)
\begin{equation}
  \begin{aligned}
    &\maxover[X,Y] \qquad &&\sum_{i=1}^nw_iX_{ii} -\eta\norm{X-Y}_F^2\\
    &\subto && X_{ii}=X_{0i}\quad\text{~for~} i=1,2,\dots,n,\\
    & && X_{ij}=0\quad\text{~for~}\{i,j\}\in E,\\
    & && X_{00}=1,~X=X^\dagger,\\
    & && Y\ge 0,~\rank(Y)\le d.
  \end{aligned}
  \label{eq:LovaszFiniteSeesawTilde}
\end{equation}
Alternatively, one can construct an alternating maximization based on 
(theoretically in the limit $\eta\to+\infty$)
\begin{equation}
  \begin{aligned}
    &\maxover[X,P] \qquad &&\sum_{i=1}^n w_iX_{ii}-\eta\Tr(PX)\\
    &\subto && X_{ii}=X_{0i}\quad\text{~for~} i=1,2,\dots,n,\\
    & && X_{ij}=0\quad\text{~for~}\{i,j\}\in E,\\
    & && X_{00}=1,~X\ge 0,\\
    & && P\ge 0,~P^2=P,~\rank(P)=n-d+1.
  \end{aligned}
  \label{eq:LovaszFiniteSeesawTildeAlt}
\end{equation}
Again, it is usually sufficient to take, for example, $\eta=\max_i\abs{w_i}$ or 
$\eta=10\max_i\abs{w_i}$ in practice. After performing the alternating 
optimization, we can confirm that maximization gives a lower bound for 
$\tilde\vartheta_d(G,\vw)$ if $\norm{X-Y}_F=0$ in 
Eq.~\eqref{eq:LovaszFiniteSeesawTilde} or $\Tr(PX)=0$ in 
Eq.~\eqref{eq:LovaszFiniteSeesawTildeAlt}.

As an example, with either Eq.~\eqref{eq:LovaszFiniteSeesawTilde} or 
Eq.~\eqref{eq:LovaszFiniteSeesawTildeAlt}, we show that 
\begin{equation}
  \tilde\vartheta_3(G_\KK)\ge3.3380.
  \label{eq:lowerGKKTilde}
\end{equation}
This can also be verified analytically by 
taking
\begin{equation}
  \begin{aligned}
    \ket{\psi_1}&=(1,0,0)\\
    \ket{\psi_2}&=(0,1,0)\\
    \ket{\psi_3}&=(0,0,1)\\
    \ket{\psi_4}&=(0,a,\sqrt{1-a^2})\\
    \ket{\psi_5}&=(\frac{1}{\sqrt{2}},0,\frac{1}{\sqrt{2}})\\
    \ket{\psi_6}&=(\sqrt{1-a^2},a,0)\\
    \ket{\psi_7}&=(\frac{a}{\sqrt{1+a^2}},\frac{\sqrt{1-a^2}}{\sqrt{1+a^2}},
		  -\frac{a}{\sqrt{1+a^2}})\\
    \ket{\psi_8}&=(0,0,0)\\
    \ket{\psi_9}&=(\frac{a}{\sqrt{1+a^2}},-\frac{\sqrt{1-a^2}}{\sqrt{1+a^2}},
		  \frac{a}{\sqrt{1+a^2}})
  \end{aligned}
\end{equation}
and $\ket{\varphi}$ is taken as the eigenvector corresponding to the largest 
eigenvalue of $\sum_{i=1}^9\ketbra{\psi_i}$. Note that we have taken 
$\ket{\psi_8}$ as a zero vector. Taking $X$ as the Gram matrix of 
$(\ket{\varphi},\braket{\psi_1}{\varphi}\ket{\psi_1},
\braket{\psi_2}{\varphi}\ket{\psi_2},
\dots,\braket{\psi_9}{\varphi}\ket{\psi_9})$, 
one can easily verify that $X$ satisfies all the constraints in
Eq.~\eqref{eq:LovaszFiniteRayRelaxApp} and
\begin{equation}
  \sum_{i=1}^9X_{ii}=\sum_{i=1}^9\abs{\braket{\varphi}{\psi_i}}^2=3.3380,
\end{equation}
when $a=0.5121$. Thus, we prove that $\tilde\vartheta_3(G_{\KK})\ge 3.3380$.

\section{Upper bounds for $\vartheta_d(G,\vw)$ and $\tilde\vartheta_d(G,\vw)$}
\label{app:upperBound}
%

In the main text, we have shown that $\vartheta_d(G,\vw)$ can be written as
\begin{equation}
  \begin{aligned}
    &\maxover[X] \qquad &&\sum_{i=1}^nw_i\abs{X_{0i}}^2\\
    &\subto && X_{ii}=1\quad\text{~for~} i=0,1,2,\dots,n,\\
    & && X_{ij}=0\quad\text{~for~}\{i,j\}\in E,\\
    & && X\ge 0,~\rank(X)\le d,
  \end{aligned}
  \label{eq:LovaszFiniteQuadApp2}
\end{equation}
where $[X_{ij}]_{i,j=0}^n$ is an $(n+1)\times (n+1)$ Hermitian matrix. By 
taking advantage of the result in Ref.~\cite{Yu.etal2022}, one can directly get 
that Eq.~\eqref{eq:LovaszFiniteQuadApp2} is equivalent to the convex program
\begin{equation}
  \begin{aligned}
    &\maxover[\Phi_{AB}]  &&
    \Tr\big[W_{AB}\Phi_{AB}\big]\\
    &\subto && 
    \Phi_{AB}\in\Sep,~V_{AB}\Phi_{AB}=\Phi_{AB},~\Tr(\Phi_{AB})=(n+1)^2,\\
    & &&
    \Tr_A[(\ketbra{i}{j}\otimes\I_d\otimes\I_B)\Phi_{AB}]=0
    \quad\text{~for~}\{i,j\}\in E,\\
    & &&
    \Tr_A[(\ketbra{i}\otimes\I_d\otimes\I_B)\Phi_{AB}]
    =\frac{1}{n+1}\Tr_A\big[\Phi_{AB}\big]
    \quad\text{~for~}i=0,1,\dots,n,
  \end{aligned}
  \label{eq:LovaszFiniteConvexApp}
\end{equation}
where $\Phi_{AB}$ is an unnormalized state in $\cH_A\otimes\cH_B$ with 
$\cH_A=\cH_B=\dC^{n+1}\otimes\dC^d$ and
\begin{equation}
  W_{AB}=\frac{1}{2}\left(\sum_{i=1}^nw_i\ketbra{0}{i}
    \otimes\I_d\otimes\ketbra{i}{0}\otimes\I_d
    + \sum_{i=1}^nw_i\ketbra{i}{0}\otimes\I_d
  \otimes\ketbra{0}{i}\otimes\I_d\right).
\end{equation}
Here the average is taken for making sure that $W_{AB}$ is Hermitian and thus 
the objective function is always real.

Theoretically, one can also construct a complete hierarchy for calculating 
$\vartheta_d(G,\vw)$ from Eq.~\eqref{eq:LovaszFiniteConvexApp}, whose lowest 
order (i.e., the PPT relaxation) can be written as (after applying the inherent 
symmetry) \cite{Yu.etal2022}
\begin{equation}
  \begin{aligned}
    &\maxover[\Phi_I,\Phi_V]  &&
    \Tr\left[
      \frac{1}{2}\sum_{i=1}^nw_i\big(\ketbra{0}{i}\otimes\ketbra{i}{0}
      +\ketbra{i}{0}\otimes\ketbra{0}{i}\big)\big(d^2\Phi_I+d\Phi_V\big)
    \right]\\
    &\subto && \Phi_I+\Phi_V\ge 0,~\Phi_I-\Phi_V\ge 0,~\Phi_I^{T_{A_1}}\ge 0,
	       ~\Phi_I^{T_{A_1}}+d\Phi_V^{T_{A_1}}\ge 0,\\
    &       && \Phi_V=V_{A_1B_1}\Phi_I,~d^2\Tr(\Phi_I)+d\Tr(\Phi_V)=(n+1)^2,\\
    &       && \Tr_{A_1}\big[(\ketbra{i}{j}\otimes\I_{n+1})
	       (d^2\Phi_I+d\Phi_V)\big]=0\quad\text{~for~}\{i,j\}\in E,\\
    &       && \Tr_{A_1}\big[(\ketbra{i}\otimes\I_{n+1})
	       (d^2\Phi_I+d\Phi_V)\big]=\frac{1}{n+1}\Tr_{A_1}
	     \big[d^2\Phi_I+d\Phi_V\big]\quad\text{~for~}i=0,1,\dots,n,
  \end{aligned}
  \label{eq:FiniteLovaszQuadRelax}
\end{equation}
where $\Phi_I$ and $\Phi_V$ are Hermitian operators on 
$\cH_{A_1}\otimes\cH_{B_1}=\dC^{n+1}\otimes\dC^{n+1}$, i.e., they are 
$(n+1)^2\times(n+1)^2$ Hermitian matrices. Actually, as all parameters in the 
SDP are real, we can take $\Phi_I$ and $\Phi_V$ as symmetric real matrices.
In practice, the upper bound from Eq.~\eqref{eq:FiniteLovaszQuadRelax} can be 
trivial. For example, the upper bound from Eq.~\eqref{eq:FiniteLovaszQuadRelax} 
for $\vartheta_3(G_\KK)$ is even larger than $\vartheta(G_\KK)=3.4704$.
In the following, we provide several alternative methods for upper bounding 
$\vartheta^d(G,\vw)$ and $\tilde\vartheta^d(G,\vw)$ based on 
Eq.~\eqref{eq:LovaszFiniteRayRelax}, which we rewrite below
\begin{equation}
  \begin{aligned}
    &\maxover[X] \qquad &&\sum_{i=1}^n w_iX_{ii}\\
    &\subto && X_{ii}=X_{0i}\quad\text{~for~} i=1,2,\dots,n,\\
    & && X_{ij}=0\quad\text{~for~}\{i,j\}\in E,\\
    & && X_{00}=1,~X\ge 0,~\rank(X)\le d,
  \end{aligned}
  \label{eq:LovaszFiniteRayRelaxApp2}
\end{equation}
where $[X_{ij}]_{i,j=0}^n$ is an $(n+1)\times (n+1)$ Hermitian matrix. Note 
that $X$ can be interpreted as the Gram matrix of 
$(\ket{\varphi},\braket{\psi_1}{\varphi}\ket{\psi_1},
\braket{\psi_2}{\varphi}\ket{\psi_2},\dots,
\braket{\psi_n}{\varphi}\ket{\psi_n})$ and thus 
$X_{ii}=X_{0i}=\abs{\braket{\psi_i}{\varphi}}^2$ for $i=1,2,\dots,n$.

First, one can derive an upper bound for $\vartheta_d(G,\vw)$ from the 
so-called basis identity,
\begin{equation}
  X_{0i_1}+X_{0i_2}+\dots+X_{0i_d}=1
  \quad\Leftrightarrow\quad
  \abs{\braket{\psi_{i_1}}{\varphi}}^2
  +\abs{\braket{\psi_{i_2}}{\varphi}}^2
  +\dots+\abs{\braket{\psi_{i_d}}{\varphi}}^2
  =1,
  \label{eq:basisIdentity}
\end{equation}
when $\{\ket{\psi_{i_1},\ket{\psi_{i_2}},\dots,\ket{\psi_{i_d}}}\}$ forms 
a basis, i.e., the set of nodes $\{i_1,i_2,\dots,i_d\}$ forms a $d$-clique in 
the language of graph theory. Taking graph $G_\KK$ and $d=3$ as an example, we 
can add the following additional constraints to 
Eq.~\eqref{eq:LovaszFiniteRayRelaxApp2}:
\begin{equation}
  X_{01}+X_{02}+X_{03}=1,~
  X_{04}+X_{07}+X_{08}=1,
\end{equation}
and obtain a nontrivial upper bound
\begin{equation}
  \vartheta_3(G_\KK)\le 3.4153.
\end{equation}

Second, we have explained in the main text that 
$\vartheta_d(G)\le\tilde\vartheta_d(G)$. Thus, an upper bound of 
$\tilde\vartheta_d(G)$ also provides an upper of $\vartheta_d(G)$.  For  
Eq.~\eqref{eq:LovaszFiniteRayRelaxApp2}, we can also construct an equivalent 
convex optimization,
\begin{equation}
  \begin{aligned}
    &\maxover[\Phi_{AB}]  &&
    \Tr\big[(\ketbra{0}\otimes\I_d\otimes
      \sum_{i=1}^nw_i\ketbra{i}\otimes\I_d)\Phi_{AB}\big]\\
    &\subto && \Phi_{AB}\in\Sep,~V_{AB}\Phi_{AB}=\Phi_{AB},\\
    & &&
    \Tr[(\ketbra{0}\otimes\I_d\otimes
    \ketbra{0}\otimes\I_d)\Phi_{AB}]=1,\\
    & &&
    \Tr_A[(\ketbra{i}{j}\otimes\I_d\otimes\I_B)\Phi_{AB}]=0
    \quad\text{~for~}\{i,j\}\in E,\\
    & &&
    \Tr_A[(\ketbra{0}{i}\otimes\I_d\otimes\I_B)\Phi_{AB}]
    =\Tr_A[(\ketbra{i}\otimes\I_d\otimes\I_B)\Phi_{AB}]
    \quad\text{~for~}i=1,2,\dots,n,
  \end{aligned}
  \label{eq:LovaszFiniteRayRelaxConvex}
\end{equation}
where $\Phi_{AB}$ is an unnormalized state in $\cH_A\otimes\cH_B$. The proof of 
the equivalence does not follow directly from Ref.~\cite{Yu.etal2022}, but 
a similar argument can be formulated. The difference is that now $\Phi_{AB}$ 
admits the form
\begin{equation}
  \Phi_{AB}=\sum_\ell p_\ell\ketbra{\chi_\ell}_A\otimes\ketbra{\chi_\ell}_B
  +\sum_s\ketbra{\xi_s}_A\otimes\ketbra{\xi_s}_B,
\end{equation}
where $p_\ell$ form a probability distribution, and
\begin{subequations}
  \begin{align}
    &\Tr[(\ketbra{0}\otimes\I_d)\ketbra{\chi_\ell}]=1,\\
    &\Tr[(\ketbra{i}{j}\otimes\I_d)\ketbra{\chi_\ell}]=0
    \quad\text{~for~}\{i,j\}\in E,\\
    &\Tr[(\ketbra{0}{i}\otimes\I_d)\ketbra{\chi_\ell}]=
    \Tr[(\ketbra{i}\otimes\I_d)\ketbra{\chi_\ell}]
    \quad\text{~for~}i=1,2,\dots,n,
  \end{align}
\end{subequations}
\begin{subequations}
  \begin{align}
    &\Tr[(\ketbra{0}\otimes\I_d)\ketbra{\xi_s}]=0,\\
    &\Tr[(\ketbra{i}{j}\otimes\I_d)\ketbra{\xi_s}]=0
    \quad\text{~for~}\{i,j\}\in E\\
    &\Tr[(\ketbra{0}{i}\otimes\I_d)\ketbra{\xi_s}]=
    \Tr[(\ketbra{i}\otimes\I_d)\ketbra{\xi_s}]
    \quad\text{~for~}i=1,2,\dots,n,
  \end{align}
\end{subequations}
from which one can easily verify that $\ketbra{\xi_s}=0$ and thus $\Phi_{AB}$ 
is of the form
\begin{equation}
  \Phi_{AB}=\sum_\ell p_\ell\ketbra{\chi_\ell}_A\otimes\ketbra{\chi_\ell}_B,
\end{equation}
where $\Tr_2(\ketbra{\chi_\ell})$ are in the feasible region of 
Eq.~\eqref{eq:LovaszFiniteRayRelaxApp2}.

Again, one can construct a complete hierarchy for calculating 
$\tilde\vartheta_d(G,\vw)$ from Eq.~\eqref{eq:LovaszFiniteRayRelaxConvex}, 
whose lowest order (i.e., the PPT relaxation) can be written as (after applying 
the inherent symmetry)
\begin{equation}
  \begin{aligned}
    &\maxover[\Phi_I,\Phi_V]  &&
    \Tr\Big[\big(\ketbra{0}\otimes\sum_{i=1}^nw_i\ketbra{i}\big)
	       \big(d^2\Phi_I+d\Phi_V\big)\Big]\\
    &\subto && \Phi_I+\Phi_V\ge 0,~\Phi_I-\Phi_V\ge 0,~\Phi_I^{T_{A_1}}\ge 0,
	       ~\Phi_I^{T_{A_1}}+d\Phi_V^{T_{A_1}}\ge 0,\\
    &       && \Phi_V=V_{A_1B_1}\Phi_I,
	       ~\Tr\big[(\ketbra{0}\otimes\ketbra{0})
	       (d^2\Phi_I+d\Phi_V)\big]=1,\\
    &       && \Tr_{A_1}\big[(\ketbra{i}{j}\otimes\I_{n+1})
	       (d^2\Phi_I+d\Phi_V)\big]=0\quad\text{~for~}\{i,j\}\in E,\\
    &       && \Tr_{A_1}\big[(\ketbra{0}{i}\otimes\I_{n+1})
	       (d^2\Phi_I+d\Phi_V)\big]=
	       \Tr_{A_1}\big[(\ketbra{i}\otimes\I_{n+1})
	       (d^2\Phi_I+d\Phi_V)\big]
	       \quad\text{~for~}i=0,1,\dots,n,
  \end{aligned}
  \label{eq:FiniteLovaszTildeRelax}
\end{equation}
where $\Phi_I$ and $\Phi_V$ are Hermitian operators on 
$\cH_{A_1}\otimes\cH_{B_1}=\dC^{n+1}\otimes\dC^{n+1}$, i.e., they are 
$(n+1)^2\times(n+1)^2$ Hermitian matrices. Actually, as all parameters in the 
SDP are real, we can take $\Phi_I$ and $\Phi_V$ as symmetric real matrices.  
With Eq.~\eqref{eq:FiniteLovaszTildeRelax}, we get an upper bound for 
$\vartheta_3(G_\KK)$ and $\tilde\vartheta_3(G_\KK)$
\begin{equation}
  \vartheta_3(G_\KK)\le\tilde\vartheta_3(G_\KK)\le3.3637,
\end{equation}
which is a nontrivial bound, because 
$\vartheta(G_\KK)=\vartheta_4(G_\KK)=3.4704$.

The upper bound for $\vartheta_d(G)$ can be further improved when considering 
the basis identity in Eq.~\eqref{eq:basisIdentity}. The basis identity implies 
the following extra constraints for Eq.~\eqref{eq:FiniteLovaszTildeRelax}:
\begin{equation}
  \sum_{\ell=1}^d\Tr_{A_1}\big[(\ketbra{i_\ell}\otimes\I_{n+1})
  (d^2\Phi_I+d\Phi_V)\big]=
  \Tr_{A_1}\big[(\ketbra{0}\otimes\I_{n+1})(d^2\Phi_I+d\Phi_V)\big],
  \label{eq:extraConsClique}
\end{equation}
when $\{i_1,i_2,\dots,i_d\}$ is a $d$-clique. By taking advantage the two 
$3$-cliques in $G_\KK$, one can obtain an improved upper bound
\begin{equation}
  \vartheta_3(G_\KK)\le3.3405,
\end{equation}
which is still larger than than the obtained lower bound 
$\vartheta_3(G_\KK)\ge3.3333$. Interestingly, this gap can be closed by taking 
advantage the obtained lower bound. More precisely, if one knows a lower bound 
$\vartheta_d(G,\vw)\ge\theta$, the following inequality can be added to 
Eq.~\eqref{eq:LovaszFiniteRayRelaxApp2}:
\begin{equation}
  \sum_{i=1}^nw_iX_{ii}\ge\theta.
\end{equation}
Note that rounding down is necessary for choosing $\theta$. This inequality 
implies the following extra constraint for 
Eq.~\eqref{eq:FiniteLovaszTildeRelax}:
\begin{equation}
  \sum_{i=1}^n\Tr_{A_1}\big[(w_i\ketbra{i}\otimes\I_{n+1})
  (d^2\Phi_I+d\Phi_V)\big]\ge\theta
  \Tr_{A_1}\big[(\ketbra{0}\otimes\I_{n+1})(d^2\Phi_I+d\Phi_V)\big],
  \label{eq:extraConsLower}
\end{equation}
Taking $\theta=3.33$ for $G_\KK$, one can already obtain the upper bound
\begin{equation}
  \vartheta_3(G_\KK)\le3.3333,
  \label{eq:upperGKK}
\end{equation}
and thus $\vartheta_3(G_\KK)=3.3333$ up to a numerical precision.

Comparing Eq.~\eqref{eq:upperGKK} with Eq.~\eqref{eq:lowerGKKTilde},
we also get that $\tilde\vartheta_3(G_\KK)>\vartheta_3(G_\KK)$, which
provides an explicit example of $\tilde\vartheta_d(G)\ne\vartheta_d(G)$.
The reason why $\tilde\vartheta_d(G,\vw)$ may be strictly larger than 
$\vartheta_d(G,\vw)$ is that the constraints in 
Eq.~\eqref{eq:LovaszFiniteRayRelaxApp2} can also be satisfied when some 
$\ket{\psi_i}=0$. This is also the reason why the proof in
Ref.~\cite{Ray.etal2021} fails; the proof there needs the obtained vectors
to be normalizable, which is only applicable to nonzero vectors.
Thus, we also need to consider all induced subgraphs 
$\{G[S]\mid S\subseteq V\}$. More precisely,
\begin{equation}
  \tilde\vartheta_d(G,\vw)=\max_{S\subseteq V}\vartheta_d(G[S],\vw[S]),
  \label{eq:LovaszFiniteTilde}
\end{equation}
where the induced subgraph $G[S]$ is defined by removing the vertices in  
$V\setminus S$ and the edges connecting to them from $G$, and $\vw[S]$ are the 
weights on the remaining vertices. Equation~\eqref{eq:LovaszFiniteTilde} also 
implies that $\tilde\vartheta_d(G,\vw)$ exists for any $d$, because there 
always exists an induced subgraph of $G$ for which a $d$-dimensional projective 
representation exists.

The upper bound in Eq.~\eqref{eq:upperGKK} is obtained by considering both 
Eq.~\eqref{eq:extraConsClique} and Eq.~\eqref{eq:extraConsLower}. When 
considering only Eq.~\eqref{eq:extraConsLower} with $\theta=3.3380$, one can 
also get an upper bound for $\tilde\vartheta_3(G_\KK)$
\begin{equation}
  \tilde\vartheta_3(G_\KK)\le 3.3535
\end{equation}
However, this upper bound is still slightly larger than the best lower bound 
$\tilde\vartheta_3(G_\KK)\ge 3.3380$ that we obtained. Although we believe that 
$\tilde\vartheta_3(G_\KK)$ equals to $3.3380$ up to a numerical precision, this 
cannot be proved with the present method. We leave this open question for 
future research.

\twocolumngrid

\bibliography{QuantumInf}

\begin{thebibliography}{60}%
\makeatletter
\providecommand \@ifxundefined [1]{%
 \@ifx{#1\undefined}
}%
\providecommand \@ifnum [1]{%
 \ifnum #1\expandafter \@firstoftwo
 \else \expandafter \@secondoftwo
 \fi
}%
\providecommand \@ifx [1]{%
 \ifx #1\expandafter \@firstoftwo
 \else \expandafter \@secondoftwo
 \fi
}%
\providecommand \natexlab [1]{#1}%
\providecommand \enquote  [1]{``#1''}%
\providecommand \bibnamefont  [1]{#1}%
\providecommand \bibfnamefont [1]{#1}%
\providecommand \citenamefont [1]{#1}%
\providecommand \href@noop [0]{\@secondoftwo}%
\providecommand \href [0]{\begingroup \@sanitize@url \@href}%
\providecommand \@href[1]{\@@startlink{#1}\@@href}%
\providecommand \@@href[1]{\endgroup#1\@@endlink}%
\providecommand \@sanitize@url [0]{\catcode `\\12\catcode `\$12\catcode
  `\&12\catcode `\#12\catcode `\^12\catcode `\_12\catcode `\%12\relax}%
\providecommand \@@startlink[1]{}%
\providecommand \@@endlink[0]{}%
\providecommand \url  [0]{\begingroup\@sanitize@url \@url }%
\providecommand \@url [1]{\endgroup\@href {#1}{\urlprefix }}%
\providecommand \urlprefix  [0]{URL }%
\providecommand \Eprint [0]{\href }%
\providecommand \doibase [0]{https://doi.org/}%
\providecommand \selectlanguage [0]{\@gobble}%
\providecommand \bibinfo  [0]{\@secondoftwo}%
\providecommand \bibfield  [0]{\@secondoftwo}%
\providecommand \translation [1]{[#1]}%
\providecommand \BibitemOpen [0]{}%
\providecommand \bibitemStop [0]{}%
\providecommand \bibitemNoStop [0]{.\EOS\space}%
\providecommand \EOS [0]{\spacefactor3000\relax}%
\providecommand \BibitemShut  [1]{\csname bibitem#1\endcsname}%
\let\auto@bib@innerbib\@empty
\bibitem [{\citenamefont {Erhard}\ \emph {et~al.}(2020)\citenamefont {Erhard},
  \citenamefont {Krenn},\ and\ \citenamefont {Zeilinger}}]{Erhard.etal2020}%
  \BibitemOpen
  \bibfield  {author} {\bibinfo {author} {\bibfnamefont {M.}~\bibnamefont
  {Erhard}}, \bibinfo {author} {\bibfnamefont {M.}~\bibnamefont {Krenn}},\ and\
  \bibinfo {author} {\bibfnamefont {A.}~\bibnamefont {Zeilinger}},\ }\bibfield
  {title} {\bibinfo {title} {Advances in high-dimensional quantum
  entanglement},\ }\href {https://doi.org/10.1038/s42254-020-0220-6} {\bibfield
   {journal} {\bibinfo  {journal} {Nat. Rev. Phys.}\ }\textbf {\bibinfo
  {volume} {2}},\ \bibinfo {pages} {365} (\bibinfo {year} {2020})}\BibitemShut
  {NoStop}%
\bibitem [{\citenamefont {Cozzolino}\ \emph {et~al.}(2019)\citenamefont
  {Cozzolino}, \citenamefont {Da~Lio}, \citenamefont {Bacco},\ and\
  \citenamefont {Oxenl\o~we}}]{Cozzolino.etal2019}%
  \BibitemOpen
  \bibfield  {author} {\bibinfo {author} {\bibfnamefont {D.}~\bibnamefont
  {Cozzolino}}, \bibinfo {author} {\bibfnamefont {B.}~\bibnamefont {Da~Lio}},
  \bibinfo {author} {\bibfnamefont {D.}~\bibnamefont {Bacco}},\ and\ \bibinfo
  {author} {\bibfnamefont {L.~K.}\ \bibnamefont {Oxenl\o~we}},\ }\bibfield
  {title} {\bibinfo {title} {High-dimensional quantum communication: Benefits,
  progress, and future challenges},\ }\href
  {https://doi.org/10.1002/qute.201900038} {\bibfield  {journal} {\bibinfo
  {journal} {Advanced Quantum Technologies}\ }\textbf {\bibinfo {volume} {2}},\
  \bibinfo {pages} {1900038} (\bibinfo {year} {2019})}\BibitemShut {NoStop}%
\bibitem [{\citenamefont {{Bavaresco}}\ \emph {et~al.}(2018)\citenamefont
  {{Bavaresco}}, \citenamefont {{Herrera Valencia}}, \citenamefont
  {{Kl{\"o}ckl}}, \citenamefont {{Pivoluska}}, \citenamefont {{Erker}},
  \citenamefont {{Friis}}, \citenamefont {{Malik}},\ and\ \citenamefont
  {{Huber}}}]{Bavaresco.etal2018}%
  \BibitemOpen
  \bibfield  {author} {\bibinfo {author} {\bibfnamefont {J.}~\bibnamefont
  {{Bavaresco}}}, \bibinfo {author} {\bibfnamefont {N.}~\bibnamefont {{Herrera
  Valencia}}}, \bibinfo {author} {\bibfnamefont {C.}~\bibnamefont
  {{Kl{\"o}ckl}}}, \bibinfo {author} {\bibfnamefont {M.}~\bibnamefont
  {{Pivoluska}}}, \bibinfo {author} {\bibfnamefont {P.}~\bibnamefont
  {{Erker}}}, \bibinfo {author} {\bibfnamefont {N.}~\bibnamefont {{Friis}}},
  \bibinfo {author} {\bibfnamefont {M.}~\bibnamefont {{Malik}}},\ and\ \bibinfo
  {author} {\bibfnamefont {M.}~\bibnamefont {{Huber}}},\ }\bibfield  {title}
  {\bibinfo {title} {Measurements in two bases are sufficient for certifying
  high-dimensional entanglement},\ }\href
  {https://doi.org/10.1038/s41567-018-0203-z} {\bibfield  {journal} {\bibinfo
  {journal} {Nat. Phys.}\ }\textbf {\bibinfo {volume} {14}},\ \bibinfo {pages}
  {1032} (\bibinfo {year} {2018})}\BibitemShut {NoStop}%
\bibitem [{\citenamefont {Ringbauer}\ \emph {et~al.}(2018)\citenamefont
  {Ringbauer}, \citenamefont {Bromley}, \citenamefont {Cianciaruso},
  \citenamefont {Lami}, \citenamefont {Lau}, \citenamefont {Adesso},
  \citenamefont {White}, \citenamefont {Fedrizzi},\ and\ \citenamefont
  {Piani}}]{Ringbauer.etal2018}%
  \BibitemOpen
  \bibfield  {author} {\bibinfo {author} {\bibfnamefont {M.}~\bibnamefont
  {Ringbauer}}, \bibinfo {author} {\bibfnamefont {T.~R.}\ \bibnamefont
  {Bromley}}, \bibinfo {author} {\bibfnamefont {M.}~\bibnamefont
  {Cianciaruso}}, \bibinfo {author} {\bibfnamefont {L.}~\bibnamefont {Lami}},
  \bibinfo {author} {\bibfnamefont {W.~Y.~S.}\ \bibnamefont {Lau}}, \bibinfo
  {author} {\bibfnamefont {G.}~\bibnamefont {Adesso}}, \bibinfo {author}
  {\bibfnamefont {A.~G.}\ \bibnamefont {White}}, \bibinfo {author}
  {\bibfnamefont {A.}~\bibnamefont {Fedrizzi}},\ and\ \bibinfo {author}
  {\bibfnamefont {M.}~\bibnamefont {Piani}},\ }\bibfield  {title} {\bibinfo
  {title} {Certification and quantification of multilevel quantum coherence},\
  }\href {https://doi.org/10.1103/PhysRevX.8.041007} {\bibfield  {journal}
  {\bibinfo  {journal} {Phys. Rev. X}\ }\textbf {\bibinfo {volume} {8}},\
  \bibinfo {pages} {041007} (\bibinfo {year} {2018})}\BibitemShut {NoStop}%
\bibitem [{\citenamefont {Navascu\'es}\ and\ \citenamefont
  {V\'ertesi}(2015)}]{Navascues.Vertesi2015}%
  \BibitemOpen
  \bibfield  {author} {\bibinfo {author} {\bibfnamefont {M.}~\bibnamefont
  {Navascu\'es}}\ and\ \bibinfo {author} {\bibfnamefont {T.}~\bibnamefont
  {V\'ertesi}},\ }\bibfield  {title} {\bibinfo {title} {Bounding the set of
  finite dimensional quantum correlations},\ }\href
  {https://doi.org/10.1103/PhysRevLett.115.020501} {\bibfield  {journal}
  {\bibinfo  {journal} {Phys. Rev. Lett.}\ }\textbf {\bibinfo {volume} {115}},\
  \bibinfo {pages} {020501} (\bibinfo {year} {2015})}\BibitemShut {NoStop}%
\bibitem [{\citenamefont {Kochen}\ and\ \citenamefont
  {Specker}(1967)}]{Kochen.Specker1967}%
  \BibitemOpen
  \bibfield  {author} {\bibinfo {author} {\bibfnamefont {S.}~\bibnamefont
  {Kochen}}\ and\ \bibinfo {author} {\bibfnamefont {E.}~\bibnamefont
  {Specker}},\ }\bibfield  {title} {\bibinfo {title} {The problem of hidden
  variables in quantum mechanics},\ }\href@noop {} {\bibfield  {journal}
  {\bibinfo  {journal} {J. Math. Mech.}\ }\textbf {\bibinfo {volume} {17}},\
  \bibinfo {pages} {59} (\bibinfo {year} {1967})}\BibitemShut {NoStop}%
\bibitem [{\citenamefont {Peres}(1993)}]{Peres1993}%
  \BibitemOpen
  \bibfield  {author} {\bibinfo {author} {\bibfnamefont {A.}~\bibnamefont
  {Peres}},\ }\href@noop {} {\emph {\bibinfo {title} {Quantum Theory: Concepts
  and Methods}}}\ (\bibinfo  {publisher} {Kluwer Academic Publishers,
  Dordrecht},\ \bibinfo {year} {1993})\BibitemShut {NoStop}%
\bibitem [{\citenamefont {Budroni}\ \emph {et~al.}(2022)\citenamefont
  {Budroni}, \citenamefont {Cabello}, \citenamefont {G\"uhne}, \citenamefont
  {Kleinmann},\ and\ \citenamefont {Larsson}}]{Budroni.etal2022}%
  \BibitemOpen
  \bibfield  {author} {\bibinfo {author} {\bibfnamefont {C.}~\bibnamefont
  {Budroni}}, \bibinfo {author} {\bibfnamefont {A.}~\bibnamefont {Cabello}},
  \bibinfo {author} {\bibfnamefont {O.}~\bibnamefont {G\"uhne}}, \bibinfo
  {author} {\bibfnamefont {M.}~\bibnamefont {Kleinmann}},\ and\ \bibinfo
  {author} {\bibfnamefont {J.-A.}\ \bibnamefont {Larsson}},\ }\bibfield
  {title} {\bibinfo {title} {{Kochen}-{Specker} contextuality},\ }\href
  {https://doi.org/10.1103/RevModPhys.94.045007} {\bibfield  {journal}
  {\bibinfo  {journal} {Rev. Mod. Phys.}\ }\textbf {\bibinfo {volume} {94}},\
  \bibinfo {pages} {045007} (\bibinfo {year} {2022})}\BibitemShut {NoStop}%
\bibitem [{\citenamefont {{Howard}}\ \emph {et~al.}(2014)\citenamefont
  {{Howard}}, \citenamefont {{Wallman}}, \citenamefont {{Veitch}},\ and\
  \citenamefont {{Emerson}}}]{Howard.etal2014}%
  \BibitemOpen
  \bibfield  {author} {\bibinfo {author} {\bibfnamefont {M.}~\bibnamefont
  {{Howard}}}, \bibinfo {author} {\bibfnamefont {J.}~\bibnamefont {{Wallman}}},
  \bibinfo {author} {\bibfnamefont {V.}~\bibnamefont {{Veitch}}},\ and\
  \bibinfo {author} {\bibfnamefont {J.}~\bibnamefont {{Emerson}}},\ }\bibfield
  {title} {\bibinfo {title} {Contextuality supplies the `magic' for quantum
  computation},\ }\href {https://doi.org/10.1038/nature13460} {\bibfield
  {journal} {\bibinfo  {journal} {Nature (London)}\ }\textbf {\bibinfo {volume}
  {510}},\ \bibinfo {pages} {351} (\bibinfo {year} {2014})}\BibitemShut
  {NoStop}%
\bibitem [{\citenamefont {Bravyi}\ \emph {et~al.}(2018)\citenamefont {Bravyi},
  \citenamefont {Gosset},\ and\ \citenamefont {K\"onig}}]{Bravyi.etal2018}%
  \BibitemOpen
  \bibfield  {author} {\bibinfo {author} {\bibfnamefont {S.}~\bibnamefont
  {Bravyi}}, \bibinfo {author} {\bibfnamefont {D.}~\bibnamefont {Gosset}},\
  and\ \bibinfo {author} {\bibfnamefont {R.}~\bibnamefont {K\"onig}},\
  }\bibfield  {title} {\bibinfo {title} {Quantum advantage with shallow
  circuits},\ }\href {https://doi.org/10.1126/science.aar3106} {\bibfield
  {journal} {\bibinfo  {journal} {Science}\ }\textbf {\bibinfo {volume}
  {362}},\ \bibinfo {pages} {308} (\bibinfo {year} {2018})}\BibitemShut
  {NoStop}%
\bibitem [{\citenamefont {{Bravyi}}\ \emph {et~al.}(2020)\citenamefont
  {{Bravyi}}, \citenamefont {{Gosset}}, \citenamefont {{K{\"o}nig}},\ and\
  \citenamefont {{Tomamichel}}}]{Bravyi.etal2020}%
  \BibitemOpen
  \bibfield  {author} {\bibinfo {author} {\bibfnamefont {S.}~\bibnamefont
  {{Bravyi}}}, \bibinfo {author} {\bibfnamefont {D.}~\bibnamefont {{Gosset}}},
  \bibinfo {author} {\bibfnamefont {R.}~\bibnamefont {{K{\"o}nig}}},\ and\
  \bibinfo {author} {\bibfnamefont {M.}~\bibnamefont {{Tomamichel}}},\
  }\bibfield  {title} {\bibinfo {title} {{Quantum advantage with noisy shallow
  circuits}},\ }\href {https://doi.org/10.1038/s41567-020-0948-z} {\bibfield
  {journal} {\bibinfo  {journal} {Nat. Phys.}\ }\textbf {\bibinfo {volume}
  {16}},\ \bibinfo {pages} {1040} (\bibinfo {year} {2020})}\BibitemShut
  {NoStop}%
\bibitem [{\citenamefont {Bechmann-Pasquinucci}\ and\ \citenamefont
  {Peres}(2000)}]{BechmannPasquinucci.Peres2000}%
  \BibitemOpen
  \bibfield  {author} {\bibinfo {author} {\bibfnamefont {H.}~\bibnamefont
  {Bechmann-Pasquinucci}}\ and\ \bibinfo {author} {\bibfnamefont
  {A.}~\bibnamefont {Peres}},\ }\bibfield  {title} {\bibinfo {title} {Quantum
  cryptography with 3-state systems},\ }\href
  {https://doi.org/10.1103/PhysRevLett.85.3313} {\bibfield  {journal} {\bibinfo
   {journal} {Phys. Rev. Lett.}\ }\textbf {\bibinfo {volume} {85}},\ \bibinfo
  {pages} {3313} (\bibinfo {year} {2000})}\BibitemShut {NoStop}%
\bibitem [{\citenamefont {{Svozil}}()}]{Svozil2009}%
  \BibitemOpen
  \bibfield  {author} {\bibinfo {author} {\bibfnamefont {K.}~\bibnamefont
  {{Svozil}}},\ }\href@noop {} {\bibinfo {title} {{Bertlmann's chocolate balls
  and quantum type cryptography}}},\ \Eprint {https://arxiv.org/abs/0903.0231}
  {arXiv:0903.0231} \BibitemShut {NoStop}%
\bibitem [{\citenamefont {Abbott}\ \emph {et~al.}(2012)\citenamefont {Abbott},
  \citenamefont {Calude}, \citenamefont {Conder},\ and\ \citenamefont
  {Svozil}}]{Abbott.etal2012}%
  \BibitemOpen
  \bibfield  {author} {\bibinfo {author} {\bibfnamefont {A.~A.}\ \bibnamefont
  {Abbott}}, \bibinfo {author} {\bibfnamefont {C.~S.}\ \bibnamefont {Calude}},
  \bibinfo {author} {\bibfnamefont {J.}~\bibnamefont {Conder}},\ and\ \bibinfo
  {author} {\bibfnamefont {K.}~\bibnamefont {Svozil}},\ }\bibfield  {title}
  {\bibinfo {title} {Strong {Kochen}-{Specker} theorem and incomputability of
  quantum randomness},\ }\href {https://doi.org/10.1103/PhysRevA.86.062109}
  {\bibfield  {journal} {\bibinfo  {journal} {Phys. Rev. A}\ }\textbf {\bibinfo
  {volume} {86}},\ \bibinfo {pages} {062109} (\bibinfo {year}
  {2012})}\BibitemShut {NoStop}%
\bibitem [{\citenamefont {Kulikov}\ \emph {et~al.}(2017)\citenamefont
  {Kulikov}, \citenamefont {Jerger}, \citenamefont {Poto\ifmmode~\check{c}\else
  \v{c}\fi{}nik}, \citenamefont {Wallraff},\ and\ \citenamefont
  {Fedorov}}]{Kulikov.etal2017}%
  \BibitemOpen
  \bibfield  {author} {\bibinfo {author} {\bibfnamefont {A.}~\bibnamefont
  {Kulikov}}, \bibinfo {author} {\bibfnamefont {M.}~\bibnamefont {Jerger}},
  \bibinfo {author} {\bibfnamefont {A.}~\bibnamefont
  {Poto\ifmmode~\check{c}\else \v{c}\fi{}nik}}, \bibinfo {author}
  {\bibfnamefont {A.}~\bibnamefont {Wallraff}},\ and\ \bibinfo {author}
  {\bibfnamefont {A.}~\bibnamefont {Fedorov}},\ }\bibfield  {title} {\bibinfo
  {title} {Realization of a quantum random generator certified with the
  {Kochen}-{Specker} theorem},\ }\href
  {https://doi.org/10.1103/PhysRevLett.119.240501} {\bibfield  {journal}
  {\bibinfo  {journal} {Phys. Rev. Lett.}\ }\textbf {\bibinfo {volume} {119}},\
  \bibinfo {pages} {240501} (\bibinfo {year} {2017})}\BibitemShut {NoStop}%
\bibitem [{\citenamefont {{Um}}\ \emph {et~al.}(2013)\citenamefont {{Um}},
  \citenamefont {{Zhang}}, \citenamefont {{Zhang}}, \citenamefont {{Wang}},
  \citenamefont {{Yangchao}}, \citenamefont {{Deng}}, \citenamefont {{Duan}},\
  and\ \citenamefont {{Kim}}}]{Um.etal2013}%
  \BibitemOpen
  \bibfield  {author} {\bibinfo {author} {\bibfnamefont {M.}~\bibnamefont
  {{Um}}}, \bibinfo {author} {\bibfnamefont {X.}~\bibnamefont {{Zhang}}},
  \bibinfo {author} {\bibfnamefont {J.}~\bibnamefont {{Zhang}}}, \bibinfo
  {author} {\bibfnamefont {Y.}~\bibnamefont {{Wang}}}, \bibinfo {author}
  {\bibfnamefont {S.}~\bibnamefont {{Yangchao}}}, \bibinfo {author}
  {\bibfnamefont {D.-L.}\ \bibnamefont {{Deng}}}, \bibinfo {author}
  {\bibfnamefont {L.-M.}\ \bibnamefont {{Duan}}},\ and\ \bibinfo {author}
  {\bibfnamefont {K.}~\bibnamefont {{Kim}}},\ }\bibfield  {title} {\bibinfo
  {title} {Experimental certification of random numbers via quantum
  contextuality},\ }\href {https://doi.org/10.1038/srep01627} {\bibfield
  {journal} {\bibinfo  {journal} {Sci. Rep.}\ }\textbf {\bibinfo {volume}
  {3}},\ \bibinfo {pages} {1627} (\bibinfo {year} {2013})}\BibitemShut
  {NoStop}%
\bibitem [{\citenamefont {Klyachko}\ \emph {et~al.}(2008)\citenamefont
  {Klyachko}, \citenamefont {Can}, \citenamefont {Binicio\u{g}lu},\ and\
  \citenamefont {Shumovsky}}]{Klyachko.etal2008}%
  \BibitemOpen
  \bibfield  {author} {\bibinfo {author} {\bibfnamefont {A.~A.}\ \bibnamefont
  {Klyachko}}, \bibinfo {author} {\bibfnamefont {M.~A.}\ \bibnamefont {Can}},
  \bibinfo {author} {\bibfnamefont {S.}~\bibnamefont {Binicio\u{g}lu}},\ and\
  \bibinfo {author} {\bibfnamefont {A.~S.}\ \bibnamefont {Shumovsky}},\
  }\bibfield  {title} {\bibinfo {title} {Simple test for hidden variables in
  spin-1 systems},\ }\href {https://doi.org/10.1103/PhysRevLett.101.020403}
  {\bibfield  {journal} {\bibinfo  {journal} {Phys. Rev. Lett.}\ }\textbf
  {\bibinfo {volume} {101}},\ \bibinfo {pages} {020403} (\bibinfo {year}
  {2008})}\BibitemShut {NoStop}%
\bibitem [{\citenamefont {Cabello}(2008)}]{Cabello2008}%
  \BibitemOpen
  \bibfield  {author} {\bibinfo {author} {\bibfnamefont {A.}~\bibnamefont
  {Cabello}},\ }\bibfield  {title} {\bibinfo {title} {Experimentally testable
  state-independent quantum contextuality},\ }\href
  {https://doi.org/10.1103/PhysRevLett.101.210401} {\bibfield  {journal}
  {\bibinfo  {journal} {Phys. Rev. Lett.}\ }\textbf {\bibinfo {volume} {101}},\
  \bibinfo {pages} {210401} (\bibinfo {year} {2008})}\BibitemShut {NoStop}%
\bibitem [{\citenamefont {Badzi\c{a}g}\ \emph {et~al.}(2009)\citenamefont
  {Badzi\c{a}g}, \citenamefont {Bengtsson}, \citenamefont {Cabello},\ and\
  \citenamefont {Pitowsky}}]{Badziag.etal2009}%
  \BibitemOpen
  \bibfield  {author} {\bibinfo {author} {\bibfnamefont {P.}~\bibnamefont
  {Badzi\c{a}g}}, \bibinfo {author} {\bibfnamefont {I.}~\bibnamefont
  {Bengtsson}}, \bibinfo {author} {\bibfnamefont {A.}~\bibnamefont {Cabello}},\
  and\ \bibinfo {author} {\bibfnamefont {I.}~\bibnamefont {Pitowsky}},\
  }\bibfield  {title} {\bibinfo {title} {Universality of state-independent
  violation of correlation inequalities for noncontextual theories},\ }\href
  {https://doi.org/10.1103/PhysRevLett.103.050401} {\bibfield  {journal}
  {\bibinfo  {journal} {Phys. Rev. Lett.}\ }\textbf {\bibinfo {volume} {103}},\
  \bibinfo {pages} {050401} (\bibinfo {year} {2009})}\BibitemShut {NoStop}%
\bibitem [{\citenamefont {Yu}\ and\ \citenamefont {Oh}(2012)}]{Yu.Oh2012}%
  \BibitemOpen
  \bibfield  {author} {\bibinfo {author} {\bibfnamefont {S.}~\bibnamefont
  {Yu}}\ and\ \bibinfo {author} {\bibfnamefont {C.~H.}\ \bibnamefont {Oh}},\
  }\bibfield  {title} {\bibinfo {title} {State-independent proof of
  {K}ochen-{S}pecker theorem with 13 rays},\ }\href
  {https://doi.org/10.1103/PhysRevLett.108.030402} {\bibfield  {journal}
  {\bibinfo  {journal} {Phys. Rev. Lett.}\ }\textbf {\bibinfo {volume} {108}},\
  \bibinfo {pages} {030402} (\bibinfo {year} {2012})}\BibitemShut {NoStop}%
\bibitem [{\citenamefont {Bell}(1966)}]{Bell1966}%
  \BibitemOpen
  \bibfield  {author} {\bibinfo {author} {\bibfnamefont {J.~S.}\ \bibnamefont
  {Bell}},\ }\bibfield  {title} {\bibinfo {title} {On the problem of hidden
  variables in quantum mechanics},\ }\href
  {https://doi.org/10.1103/RevModPhys.38.447} {\bibfield  {journal} {\bibinfo
  {journal} {Rev. Mod. Phys.}\ }\textbf {\bibinfo {volume} {38}},\ \bibinfo
  {pages} {447} (\bibinfo {year} {1966})}\BibitemShut {NoStop}%
\bibitem [{\citenamefont {Ramanathan}\ and\ \citenamefont
  {Horodecki}(2014)}]{Ramanathan.Horodecki2014}%
  \BibitemOpen
  \bibfield  {author} {\bibinfo {author} {\bibfnamefont {R.}~\bibnamefont
  {Ramanathan}}\ and\ \bibinfo {author} {\bibfnamefont {P.}~\bibnamefont
  {Horodecki}},\ }\bibfield  {title} {\bibinfo {title} {Necessary and
  sufficient condition for state-independent contextual measurement
  scenarios},\ }\href {https://doi.org/10.1103/PhysRevLett.112.040404}
  {\bibfield  {journal} {\bibinfo  {journal} {Phys. Rev. Lett.}\ }\textbf
  {\bibinfo {volume} {112}},\ \bibinfo {pages} {040404} (\bibinfo {year}
  {2014})}\BibitemShut {NoStop}%
\bibitem [{\citenamefont {Cabello}\ \emph {et~al.}(2015)\citenamefont
  {Cabello}, \citenamefont {Kleinmann},\ and\ \citenamefont
  {Budroni}}]{Cabello.etal2015}%
  \BibitemOpen
  \bibfield  {author} {\bibinfo {author} {\bibfnamefont {A.}~\bibnamefont
  {Cabello}}, \bibinfo {author} {\bibfnamefont {M.}~\bibnamefont {Kleinmann}},\
  and\ \bibinfo {author} {\bibfnamefont {C.}~\bibnamefont {Budroni}},\
  }\bibfield  {title} {\bibinfo {title} {Necessary and sufficient condition for
  quantum state-independent contextuality},\ }\href
  {https://doi.org/10.1103/PhysRevLett.114.250402} {\bibfield  {journal}
  {\bibinfo  {journal} {Phys. Rev. Lett.}\ }\textbf {\bibinfo {volume} {114}},\
  \bibinfo {pages} {250402} (\bibinfo {year} {2015})}\BibitemShut {NoStop}%
\bibitem [{\citenamefont {G\"uhne}\ \emph {et~al.}(2014)\citenamefont
  {G\"uhne}, \citenamefont {Budroni}, \citenamefont {Cabello}, \citenamefont
  {Kleinmann},\ and\ \citenamefont {Larsson}}]{Guehne.etal2014}%
  \BibitemOpen
  \bibfield  {author} {\bibinfo {author} {\bibfnamefont {O.}~\bibnamefont
  {G\"uhne}}, \bibinfo {author} {\bibfnamefont {C.}~\bibnamefont {Budroni}},
  \bibinfo {author} {\bibfnamefont {A.}~\bibnamefont {Cabello}}, \bibinfo
  {author} {\bibfnamefont {M.}~\bibnamefont {Kleinmann}},\ and\ \bibinfo
  {author} {\bibfnamefont {J.-A.}\ \bibnamefont {Larsson}},\ }\bibfield
  {title} {\bibinfo {title} {Bounding the quantum dimension with
  contextuality},\ }\href {https://doi.org/10.1103/PhysRevA.89.062107}
  {\bibfield  {journal} {\bibinfo  {journal} {Phys. Rev. A}\ }\textbf {\bibinfo
  {volume} {89}},\ \bibinfo {pages} {062107} (\bibinfo {year}
  {2014})}\BibitemShut {NoStop}%
\bibitem [{\citenamefont {Ray}\ \emph {et~al.}(2021)\citenamefont {Ray},
  \citenamefont {Boddu}, \citenamefont {Bharti}, \citenamefont {Kwek},\ and\
  \citenamefont {Cabello}}]{Ray.etal2021}%
  \BibitemOpen
  \bibfield  {author} {\bibinfo {author} {\bibfnamefont {M.}~\bibnamefont
  {Ray}}, \bibinfo {author} {\bibfnamefont {N.~G.}\ \bibnamefont {Boddu}},
  \bibinfo {author} {\bibfnamefont {K.}~\bibnamefont {Bharti}}, \bibinfo
  {author} {\bibfnamefont {L.-C.}\ \bibnamefont {Kwek}},\ and\ \bibinfo
  {author} {\bibfnamefont {A.}~\bibnamefont {Cabello}},\ }\bibfield  {title}
  {\bibinfo {title} {Graph-theoretic approach to dimension witnessing},\ }\href
  {https://doi.org/10.1088/1367-2630/abcacd} {\bibfield  {journal} {\bibinfo
  {journal} {New J. Phys.}\ }\textbf {\bibinfo {volume} {23}},\ \bibinfo
  {pages} {033006} (\bibinfo {year} {2021})}\BibitemShut {NoStop}%
\bibitem [{\citenamefont {Kurzy\'nski}\ and\ \citenamefont
  {Kaszlikowski}(2012)}]{Kurzynski.Kaszlikowski2012}%
  \BibitemOpen
  \bibfield  {author} {\bibinfo {author} {\bibfnamefont {P.}~\bibnamefont
  {Kurzy\'nski}}\ and\ \bibinfo {author} {\bibfnamefont {D.}~\bibnamefont
  {Kaszlikowski}},\ }\bibfield  {title} {\bibinfo {title} {Contextuality of
  almost all qutrit states can be revealed with nine observables},\ }\href
  {https://doi.org/10.1103/PhysRevA.86.042125} {\bibfield  {journal} {\bibinfo
  {journal} {Phys. Rev. A}\ }\textbf {\bibinfo {volume} {86}},\ \bibinfo
  {pages} {042125} (\bibinfo {year} {2012})}\BibitemShut {NoStop}%
\bibitem [{\citenamefont {Cabello}\ \emph {et~al.}(2014)\citenamefont
  {Cabello}, \citenamefont {Severini},\ and\ \citenamefont
  {Winter}}]{Cabello.etal2014}%
  \BibitemOpen
  \bibfield  {author} {\bibinfo {author} {\bibfnamefont {A.}~\bibnamefont
  {Cabello}}, \bibinfo {author} {\bibfnamefont {S.}~\bibnamefont {Severini}},\
  and\ \bibinfo {author} {\bibfnamefont {A.}~\bibnamefont {Winter}},\
  }\bibfield  {title} {\bibinfo {title} {Graph-theoretic approach to quantum
  correlations},\ }\href {https://doi.org/10.1103/PhysRevLett.112.040401}
  {\bibfield  {journal} {\bibinfo  {journal} {Phys. Rev. Lett.}\ }\textbf
  {\bibinfo {volume} {112}},\ \bibinfo {pages} {040401} (\bibinfo {year}
  {2014})}\BibitemShut {NoStop}%
\bibitem [{\citenamefont {Lov{\'a}sz}(1979)}]{Lovasz1979}%
  \BibitemOpen
  \bibfield  {author} {\bibinfo {author} {\bibfnamefont {L.}~\bibnamefont
  {Lov{\'a}sz}},\ }\bibfield  {title} {\bibinfo {title} {On the {Shannon}
  capacity of a graph},\ }\href {https://doi.org/10.1109/TIT.1979.1055985}
  {\bibfield  {journal} {\bibinfo  {journal} {IEEE Trans. Inf. Theory}\
  }\textbf {\bibinfo {volume} {25}},\ \bibinfo {pages} {1} (\bibinfo {year}
  {1979})}\BibitemShut {NoStop}%
\bibitem [{Note1()}]{Note1}%
  \BibitemOpen
  \bibinfo {note} {In graph theory, $(\ket {\psi _1},\ket {\psi _2},\protect
  \dots ,\ket {\psi _n})$ is called an orthonormal representation of the
  complement graph $\protect \overline {G}$.}\BibitemShut {Stop}%
\bibitem [{\citenamefont {Yu}\ and\ \citenamefont {Tong}(2014)}]{Yu.Tong2014}%
  \BibitemOpen
  \bibfield  {author} {\bibinfo {author} {\bibfnamefont {X.-D.}\ \bibnamefont
  {Yu}}\ and\ \bibinfo {author} {\bibfnamefont {D.~M.}\ \bibnamefont {Tong}},\
  }\bibfield  {title} {\bibinfo {title} {Coexistence of {Kochen}-{Specker}
  inequalities and noncontextuality inequalities},\ }\href
  {https://doi.org/10.1103/PhysRevA.89.010101} {\bibfield  {journal} {\bibinfo
  {journal} {Phys. Rev. A}\ }\textbf {\bibinfo {volume} {89}},\ \bibinfo
  {pages} {010101} (\bibinfo {year} {2014})}\BibitemShut {NoStop}%
\bibitem [{\citenamefont {Cabello}(2016)}]{Cabello2016}%
  \BibitemOpen
  \bibfield  {author} {\bibinfo {author} {\bibfnamefont {A.}~\bibnamefont
  {Cabello}},\ }\bibfield  {title} {\bibinfo {title} {Simple method for
  experimentally testing any form of quantum contextuality},\ }\href
  {https://doi.org/10.1103/PhysRevA.93.032102} {\bibfield  {journal} {\bibinfo
  {journal} {Phys. Rev. A}\ }\textbf {\bibinfo {volume} {93}},\ \bibinfo
  {pages} {032102} (\bibinfo {year} {2016})}\BibitemShut {NoStop}%
\bibitem [{\citenamefont {Gr\"otschel}\ \emph {et~al.}(1986)\citenamefont
  {Gr\"otschel}, \citenamefont {Lov\'asz},\ and\ \citenamefont
  {Schrijver}}]{Groetschel.etal1986}%
  \BibitemOpen
  \bibfield  {author} {\bibinfo {author} {\bibfnamefont {M.}~\bibnamefont
  {Gr\"otschel}}, \bibinfo {author} {\bibfnamefont {L.}~\bibnamefont
  {Lov\'asz}},\ and\ \bibinfo {author} {\bibfnamefont {A.}~\bibnamefont
  {Schrijver}},\ }\bibfield  {title} {\bibinfo {title} {Relaxations of vertex
  packing},\ }\href {https://doi.org/10.1016/0095-8956(86)90087-0} {\bibfield
  {journal} {\bibinfo  {journal} {J. Combin. Theory B}\ }\textbf {\bibinfo
  {volume} {40}},\ \bibinfo {pages} {330} (\bibinfo {year} {1986})}\BibitemShut
  {NoStop}%
\bibitem [{\citenamefont {Yu}\ \emph {et~al.}(2021)\citenamefont {Yu},
  \citenamefont {Simnacher}, \citenamefont {Wyderka}, \citenamefont {Nguyen},\
  and\ \citenamefont {G{\"u}hne}}]{Yu.etal2021}%
  \BibitemOpen
  \bibfield  {author} {\bibinfo {author} {\bibfnamefont {X.-D.}\ \bibnamefont
  {Yu}}, \bibinfo {author} {\bibfnamefont {T.}~\bibnamefont {Simnacher}},
  \bibinfo {author} {\bibfnamefont {N.}~\bibnamefont {Wyderka}}, \bibinfo
  {author} {\bibfnamefont {H.~C.}\ \bibnamefont {Nguyen}},\ and\ \bibinfo
  {author} {\bibfnamefont {O.}~\bibnamefont {G{\"u}hne}},\ }\bibfield  {title}
  {\bibinfo {title} {A complete hierarchy for the pure state marginal problem
  in quantum mechanics},\ }\href {https://doi.org/10.1038/s41467-020-20799-5}
  {\bibfield  {journal} {\bibinfo  {journal} {Nat. Commun.}\ }\textbf {\bibinfo
  {volume} {12}},\ \bibinfo {pages} {1012} (\bibinfo {year}
  {2021})}\BibitemShut {NoStop}%
\bibitem [{\citenamefont {Yu}\ \emph {et~al.}(2022)\citenamefont {Yu},
  \citenamefont {Simnacher}, \citenamefont {Nguyen},\ and\ \citenamefont
  {G\"uhne}}]{Yu.etal2022}%
  \BibitemOpen
  \bibfield  {author} {\bibinfo {author} {\bibfnamefont {X.-D.}\ \bibnamefont
  {Yu}}, \bibinfo {author} {\bibfnamefont {T.}~\bibnamefont {Simnacher}},
  \bibinfo {author} {\bibfnamefont {H.~C.}\ \bibnamefont {Nguyen}},\ and\
  \bibinfo {author} {\bibfnamefont {O.}~\bibnamefont {G\"uhne}},\ }\bibfield
  {title} {\bibinfo {title} {Quantum-inspired hierarchy for rank-constrained
  optimization},\ }\href {https://doi.org/10.1103/PRXQuantum.3.010340}
  {\bibfield  {journal} {\bibinfo  {journal} {PRX Quantum}\ }\textbf {\bibinfo
  {volume} {3}},\ \bibinfo {pages} {010340} (\bibinfo {year}
  {2022})}\BibitemShut {NoStop}%
\bibitem [{\citenamefont {Davenport}\ and\ \citenamefont
  {Romberg}(2016)}]{Davenport.Romberg2016}%
  \BibitemOpen
  \bibfield  {author} {\bibinfo {author} {\bibfnamefont {M.~A.}\ \bibnamefont
  {Davenport}}\ and\ \bibinfo {author} {\bibfnamefont {J.}~\bibnamefont
  {Romberg}},\ }\bibfield  {title} {\bibinfo {title} {An overview of low-rank
  matrix recovery from incomplete observations},\ }\href
  {https://doi.org/10.1109/JSTSP.2016.2539100} {\bibfield  {journal} {\bibinfo
  {journal} {IEEE J. Sel. Top. Signal Process.}\ }\textbf {\bibinfo {volume}
  {10}},\ \bibinfo {pages} {608} (\bibinfo {year} {2016})}\BibitemShut
  {NoStop}%
\bibitem [{\citenamefont {Bowles}\ \emph {et~al.}(2014)\citenamefont {Bowles},
  \citenamefont {Quintino},\ and\ \citenamefont {Brunner}}]{Bowles.etal2014}%
  \BibitemOpen
  \bibfield  {author} {\bibinfo {author} {\bibfnamefont {J.}~\bibnamefont
  {Bowles}}, \bibinfo {author} {\bibfnamefont {M.~T.}\ \bibnamefont
  {Quintino}},\ and\ \bibinfo {author} {\bibfnamefont {N.}~\bibnamefont
  {Brunner}},\ }\bibfield  {title} {\bibinfo {title} {Certifying the dimension
  of classical and quantum systems in a prepare-and-measure scenario with
  independent devices},\ }\href
  {https://doi.org/10.1103/PhysRevLett.112.140407} {\bibfield  {journal}
  {\bibinfo  {journal} {Phys. Rev. Lett.}\ }\textbf {\bibinfo {volume} {112}},\
  \bibinfo {pages} {140407} (\bibinfo {year} {2014})}\BibitemShut {NoStop}%
\bibitem [{\citenamefont {Donohue}\ and\ \citenamefont
  {Wolfe}(2015)}]{Donohue.Wolfe2015}%
  \BibitemOpen
  \bibfield  {author} {\bibinfo {author} {\bibfnamefont {J.~M.}\ \bibnamefont
  {Donohue}}\ and\ \bibinfo {author} {\bibfnamefont {E.}~\bibnamefont
  {Wolfe}},\ }\bibfield  {title} {\bibinfo {title} {Identifying nonconvexity in
  the sets of limited-dimension quantum correlations},\ }\href
  {https://doi.org/10.1103/PhysRevA.92.062120} {\bibfield  {journal} {\bibinfo
  {journal} {Phys. Rev. A}\ }\textbf {\bibinfo {volume} {92}},\ \bibinfo
  {pages} {062120} (\bibinfo {year} {2015})}\BibitemShut {NoStop}%
\bibitem [{\citenamefont {Sikora}\ \emph {et~al.}(2016)\citenamefont {Sikora},
  \citenamefont {Varvitsiotis},\ and\ \citenamefont {Wei}}]{Sikora.etal2016}%
  \BibitemOpen
  \bibfield  {author} {\bibinfo {author} {\bibfnamefont {J.}~\bibnamefont
  {Sikora}}, \bibinfo {author} {\bibfnamefont {A.}~\bibnamefont
  {Varvitsiotis}},\ and\ \bibinfo {author} {\bibfnamefont {Z.}~\bibnamefont
  {Wei}},\ }\bibfield  {title} {\bibinfo {title} {Minimum dimension of a
  {H}ilbert space needed to generate a quantum correlation},\ }\href
  {https://doi.org/10.1103/PhysRevLett.117.060401} {\bibfield  {journal}
  {\bibinfo  {journal} {Phys. Rev. Lett.}\ }\textbf {\bibinfo {volume} {117}},\
  \bibinfo {pages} {060401} (\bibinfo {year} {2016})}\BibitemShut {NoStop}%
\bibitem [{\citenamefont {Mao}\ \emph {et~al.}(2022)\citenamefont {Mao},
  \citenamefont {Spee}, \citenamefont {Xu},\ and\ \citenamefont
  {G\"uhne}}]{Mao.etal2022}%
  \BibitemOpen
  \bibfield  {author} {\bibinfo {author} {\bibfnamefont {Y.}~\bibnamefont
  {Mao}}, \bibinfo {author} {\bibfnamefont {C.}~\bibnamefont {Spee}}, \bibinfo
  {author} {\bibfnamefont {Z.-P.}\ \bibnamefont {Xu}},\ and\ \bibinfo {author}
  {\bibfnamefont {O.}~\bibnamefont {G\"uhne}},\ }\bibfield  {title} {\bibinfo
  {title} {Structure of dimension-bounded temporal correlations},\ }\href
  {https://doi.org/10.1103/PhysRevA.105.L020201} {\bibfield  {journal}
  {\bibinfo  {journal} {Phys. Rev. A}\ }\textbf {\bibinfo {volume} {105}},\
  \bibinfo {pages} {L020201} (\bibinfo {year} {2022})}\BibitemShut {NoStop}%
\bibitem [{\citenamefont {Ali-Khan}\ \emph {et~al.}(2007)\citenamefont
  {Ali-Khan}, \citenamefont {Broadbent},\ and\ \citenamefont
  {Howell}}]{Ali-Khan.etal2007}%
  \BibitemOpen
  \bibfield  {author} {\bibinfo {author} {\bibfnamefont {I.}~\bibnamefont
  {Ali-Khan}}, \bibinfo {author} {\bibfnamefont {C.~J.}\ \bibnamefont
  {Broadbent}},\ and\ \bibinfo {author} {\bibfnamefont {J.~C.}\ \bibnamefont
  {Howell}},\ }\bibfield  {title} {\bibinfo {title} {Large-alphabet quantum key
  distribution using energy-time entangled bipartite states},\ }\href
  {https://doi.org/10.1103/PhysRevLett.98.060503} {\bibfield  {journal}
  {\bibinfo  {journal} {Phys. Rev. Lett.}\ }\textbf {\bibinfo {volume} {98}},\
  \bibinfo {pages} {060503} (\bibinfo {year} {2007})}\BibitemShut {NoStop}%
\bibitem [{\citenamefont {Neeley}\ \emph {et~al.}(2009)\citenamefont {Neeley},
  \citenamefont {Ansmann}, \citenamefont {Bialczak}, \citenamefont {Hofheinz},
  \citenamefont {Lucero}, \citenamefont {O{\textquoteright}Connell},
  \citenamefont {Sank}, \citenamefont {Wang}, \citenamefont {Wenner},
  \citenamefont {Cleland}, \citenamefont {Geller},\ and\ \citenamefont
  {Martinis}}]{Neeley.etal2009}%
  \BibitemOpen
  \bibfield  {author} {\bibinfo {author} {\bibfnamefont {M.}~\bibnamefont
  {Neeley}}, \bibinfo {author} {\bibfnamefont {M.}~\bibnamefont {Ansmann}},
  \bibinfo {author} {\bibfnamefont {R.~C.}\ \bibnamefont {Bialczak}}, \bibinfo
  {author} {\bibfnamefont {M.}~\bibnamefont {Hofheinz}}, \bibinfo {author}
  {\bibfnamefont {E.}~\bibnamefont {Lucero}}, \bibinfo {author} {\bibfnamefont
  {A.~D.}\ \bibnamefont {O{\textquoteright}Connell}}, \bibinfo {author}
  {\bibfnamefont {D.}~\bibnamefont {Sank}}, \bibinfo {author} {\bibfnamefont
  {H.}~\bibnamefont {Wang}}, \bibinfo {author} {\bibfnamefont {J.}~\bibnamefont
  {Wenner}}, \bibinfo {author} {\bibfnamefont {A.~N.}\ \bibnamefont {Cleland}},
  \bibinfo {author} {\bibfnamefont {M.~R.}\ \bibnamefont {Geller}},\ and\
  \bibinfo {author} {\bibfnamefont {J.~M.}\ \bibnamefont {Martinis}},\
  }\bibfield  {title} {\bibinfo {title} {Emulation of a quantum spin with a
  superconducting phase qudit},\ }\href
  {https://doi.org/10.1126/science.1173440} {\bibfield  {journal} {\bibinfo
  {journal} {Science}\ }\textbf {\bibinfo {volume} {325}},\ \bibinfo {pages}
  {722} (\bibinfo {year} {2009})}\BibitemShut {NoStop}%
\bibitem [{\citenamefont {{Lanyon}}\ \emph {et~al.}(2009)\citenamefont
  {{Lanyon}}, \citenamefont {{Barbieri}}, \citenamefont {{Almeida}},
  \citenamefont {{Jennewein}}, \citenamefont {{Ralph}}, \citenamefont
  {{Resch}}, \citenamefont {{Pryde}}, \citenamefont {{O'Brien}}, \citenamefont
  {{Gilchrist}},\ and\ \citenamefont {{White}}}]{Lanyon.etal2009}%
  \BibitemOpen
  \bibfield  {author} {\bibinfo {author} {\bibfnamefont {B.~P.}\ \bibnamefont
  {{Lanyon}}}, \bibinfo {author} {\bibfnamefont {M.}~\bibnamefont
  {{Barbieri}}}, \bibinfo {author} {\bibfnamefont {M.~P.}\ \bibnamefont
  {{Almeida}}}, \bibinfo {author} {\bibfnamefont {T.}~\bibnamefont
  {{Jennewein}}}, \bibinfo {author} {\bibfnamefont {T.~C.}\ \bibnamefont
  {{Ralph}}}, \bibinfo {author} {\bibfnamefont {K.~J.}\ \bibnamefont
  {{Resch}}}, \bibinfo {author} {\bibfnamefont {G.~J.}\ \bibnamefont
  {{Pryde}}}, \bibinfo {author} {\bibfnamefont {J.~L.}\ \bibnamefont
  {{O'Brien}}}, \bibinfo {author} {\bibfnamefont {A.}~\bibnamefont
  {{Gilchrist}}},\ and\ \bibinfo {author} {\bibfnamefont {A.~G.}\ \bibnamefont
  {{White}}},\ }\bibfield  {title} {\bibinfo {title} {Simplifying quantum logic
  using higher-dimensional {H}ilbert spaces},\ }\href
  {https://doi.org/10.1038/nphys1150} {\bibfield  {journal} {\bibinfo
  {journal} {Nat. Phys.}\ }\textbf {\bibinfo {volume} {5}},\ \bibinfo {pages}
  {134} (\bibinfo {year} {2009})}\BibitemShut {NoStop}%
\bibitem [{\citenamefont {{Dada}}\ \emph {et~al.}(2011)\citenamefont {{Dada}},
  \citenamefont {{Leach}}, \citenamefont {{Buller}}, \citenamefont
  {{Padgett}},\ and\ \citenamefont {{Andersson}}}]{Dada.etal2011}%
  \BibitemOpen
  \bibfield  {author} {\bibinfo {author} {\bibfnamefont {A.~C.}\ \bibnamefont
  {{Dada}}}, \bibinfo {author} {\bibfnamefont {J.}~\bibnamefont {{Leach}}},
  \bibinfo {author} {\bibfnamefont {G.~S.}\ \bibnamefont {{Buller}}}, \bibinfo
  {author} {\bibfnamefont {M.~J.}\ \bibnamefont {{Padgett}}},\ and\ \bibinfo
  {author} {\bibfnamefont {E.}~\bibnamefont {{Andersson}}},\ }\bibfield
  {title} {\bibinfo {title} {Experimental high-dimensional two-photon
  entanglement and violations of generalized {Bell} inequalities},\ }\href
  {https://doi.org/10.1038/nphys1996} {\bibfield  {journal} {\bibinfo
  {journal} {Nat. Phys.}\ }\textbf {\bibinfo {volume} {7}},\ \bibinfo {pages}
  {677} (\bibinfo {year} {2011})}\BibitemShut {NoStop}%
\bibitem [{\citenamefont {{Kues}}\ \emph {et~al.}(2017)\citenamefont {{Kues}},
  \citenamefont {{Reimer}}, \citenamefont {{Roztocki}}, \citenamefont
  {{Cort{\'e}s}}, \citenamefont {{Sciara}}, \citenamefont {{Wetzel}},
  \citenamefont {{Zhang}}, \citenamefont {{Cino}}, \citenamefont {{Chu}},\ and\
  \citenamefont {{Little}}}]{Kues.etal2017}%
  \BibitemOpen
  \bibfield  {author} {\bibinfo {author} {\bibfnamefont {M.}~\bibnamefont
  {{Kues}}}, \bibinfo {author} {\bibfnamefont {C.}~\bibnamefont {{Reimer}}},
  \bibinfo {author} {\bibfnamefont {P.}~\bibnamefont {{Roztocki}}}, \bibinfo
  {author} {\bibfnamefont {L.~R.}\ \bibnamefont {{Cort{\'e}s}}}, \bibinfo
  {author} {\bibfnamefont {S.}~\bibnamefont {{Sciara}}}, \bibinfo {author}
  {\bibfnamefont {B.}~\bibnamefont {{Wetzel}}}, \bibinfo {author}
  {\bibfnamefont {Y.}~\bibnamefont {{Zhang}}}, \bibinfo {author} {\bibfnamefont
  {A.}~\bibnamefont {{Cino}}}, \bibinfo {author} {\bibfnamefont {S.~T.}\
  \bibnamefont {{Chu}}},\ and\ \bibinfo {author} {\bibfnamefont {B.~E.}\
  \bibnamefont {{Little}}},\ }\bibfield  {title} {\bibinfo {title} {On-chip
  generation of high-dimensional entangled quantum states and their coherent
  control},\ }\href {https://doi.org/10.1038/nature22986} {\bibfield  {journal}
  {\bibinfo  {journal} {Nature}\ }\textbf {\bibinfo {volume} {546}},\ \bibinfo
  {pages} {622} (\bibinfo {year} {2017})}\BibitemShut {NoStop}%
\bibitem [{\citenamefont {Wang}\ \emph {et~al.}(2018)\citenamefont {Wang},
  \citenamefont {Paesani}, \citenamefont {Ding}, \citenamefont {Santagati},
  \citenamefont {Skrzypczyk}, \citenamefont {Salavrakos}, \citenamefont {Tura},
  \citenamefont {Augusiak}, \citenamefont {Man{\v c}inska}, \citenamefont
  {Bacco}, \citenamefont {Bonneau}, \citenamefont {Silverstone}, \citenamefont
  {Gong}, \citenamefont {Ac{\'\i}n}, \citenamefont {Rottwitt}, \citenamefont
  {Oxenl{\o}we}, \citenamefont {O{\textquoteright}Brien}, \citenamefont
  {Laing},\ and\ \citenamefont {Thompson}}]{Wang.etal2018}%
  \BibitemOpen
  \bibfield  {author} {\bibinfo {author} {\bibfnamefont {J.}~\bibnamefont
  {Wang}}, \bibinfo {author} {\bibfnamefont {S.}~\bibnamefont {Paesani}},
  \bibinfo {author} {\bibfnamefont {Y.}~\bibnamefont {Ding}}, \bibinfo {author}
  {\bibfnamefont {R.}~\bibnamefont {Santagati}}, \bibinfo {author}
  {\bibfnamefont {P.}~\bibnamefont {Skrzypczyk}}, \bibinfo {author}
  {\bibfnamefont {A.}~\bibnamefont {Salavrakos}}, \bibinfo {author}
  {\bibfnamefont {J.}~\bibnamefont {Tura}}, \bibinfo {author} {\bibfnamefont
  {R.}~\bibnamefont {Augusiak}}, \bibinfo {author} {\bibfnamefont
  {L.}~\bibnamefont {Man{\v c}inska}}, \bibinfo {author} {\bibfnamefont
  {D.}~\bibnamefont {Bacco}}, \bibinfo {author} {\bibfnamefont
  {D.}~\bibnamefont {Bonneau}}, \bibinfo {author} {\bibfnamefont {J.~W.}\
  \bibnamefont {Silverstone}}, \bibinfo {author} {\bibfnamefont
  {Q.}~\bibnamefont {Gong}}, \bibinfo {author} {\bibfnamefont {A.}~\bibnamefont
  {Ac{\'\i}n}}, \bibinfo {author} {\bibfnamefont {K.}~\bibnamefont {Rottwitt}},
  \bibinfo {author} {\bibfnamefont {L.~K.}\ \bibnamefont {Oxenl{\o}we}},
  \bibinfo {author} {\bibfnamefont {J.~L.}\ \bibnamefont
  {O{\textquoteright}Brien}}, \bibinfo {author} {\bibfnamefont
  {A.}~\bibnamefont {Laing}},\ and\ \bibinfo {author} {\bibfnamefont {M.~G.}\
  \bibnamefont {Thompson}},\ }\bibfield  {title} {\bibinfo {title}
  {Multidimensional quantum entanglement with large-scale integrated optics},\
  }\href {https://doi.org/10.1126/science.aar7053} {\bibfield  {journal}
  {\bibinfo  {journal} {Science}\ }\textbf {\bibinfo {volume} {360}},\ \bibinfo
  {pages} {285} (\bibinfo {year} {2018})}\BibitemShut {NoStop}%
\bibitem [{\citenamefont {Ecker}\ \emph {et~al.}(2019)\citenamefont {Ecker},
  \citenamefont {Bouchard}, \citenamefont {Bulla}, \citenamefont {Brandt},
  \citenamefont {Kohout}, \citenamefont {Steinlechner}, \citenamefont
  {Fickler}, \citenamefont {Malik}, \citenamefont {Guryanova}, \citenamefont
  {Ursin},\ and\ \citenamefont {Huber}}]{Ecker.etal2019}%
  \BibitemOpen
  \bibfield  {author} {\bibinfo {author} {\bibfnamefont {S.}~\bibnamefont
  {Ecker}}, \bibinfo {author} {\bibfnamefont {F.}~\bibnamefont {Bouchard}},
  \bibinfo {author} {\bibfnamefont {L.}~\bibnamefont {Bulla}}, \bibinfo
  {author} {\bibfnamefont {F.}~\bibnamefont {Brandt}}, \bibinfo {author}
  {\bibfnamefont {O.}~\bibnamefont {Kohout}}, \bibinfo {author} {\bibfnamefont
  {F.}~\bibnamefont {Steinlechner}}, \bibinfo {author} {\bibfnamefont
  {R.}~\bibnamefont {Fickler}}, \bibinfo {author} {\bibfnamefont
  {M.}~\bibnamefont {Malik}}, \bibinfo {author} {\bibfnamefont
  {Y.}~\bibnamefont {Guryanova}}, \bibinfo {author} {\bibfnamefont
  {R.}~\bibnamefont {Ursin}},\ and\ \bibinfo {author} {\bibfnamefont
  {M.}~\bibnamefont {Huber}},\ }\bibfield  {title} {\bibinfo {title}
  {Overcoming noise in entanglement distribution},\ }\href
  {https://doi.org/10.1103/PhysRevX.9.041042} {\bibfield  {journal} {\bibinfo
  {journal} {Phys. Rev. X}\ }\textbf {\bibinfo {volume} {9}},\ \bibinfo {pages}
  {041042} (\bibinfo {year} {2019})}\BibitemShut {NoStop}%
\bibitem [{\citenamefont {Kong}\ \emph {et~al.}(2023)\citenamefont {Kong},
  \citenamefont {Sun}, \citenamefont {Zhang}, \citenamefont {Zhang},\ and\
  \citenamefont {Zhang}}]{Kong.etal2023}%
  \BibitemOpen
  \bibfield  {author} {\bibinfo {author} {\bibfnamefont {L.-J.}\ \bibnamefont
  {Kong}}, \bibinfo {author} {\bibfnamefont {Y.}~\bibnamefont {Sun}}, \bibinfo
  {author} {\bibfnamefont {F.}~\bibnamefont {Zhang}}, \bibinfo {author}
  {\bibfnamefont {J.}~\bibnamefont {Zhang}},\ and\ \bibinfo {author}
  {\bibfnamefont {X.}~\bibnamefont {Zhang}},\ }\bibfield  {title} {\bibinfo
  {title} {High-dimensional entanglement-enabled holography},\ }\href
  {https://doi.org/10.1103/PhysRevLett.130.053602} {\bibfield  {journal}
  {\bibinfo  {journal} {Phys. Rev. Lett.}\ }\textbf {\bibinfo {volume} {130}},\
  \bibinfo {pages} {053602} (\bibinfo {year} {2023})}\BibitemShut {NoStop}%
\bibitem [{\citenamefont {Brunner}\ \emph {et~al.}(2008)\citenamefont
  {Brunner}, \citenamefont {Pironio}, \citenamefont {Acin}, \citenamefont
  {Gisin}, \citenamefont {M\'ethot},\ and\ \citenamefont
  {Scarani}}]{Brunner.etal2008}%
  \BibitemOpen
  \bibfield  {author} {\bibinfo {author} {\bibfnamefont {N.}~\bibnamefont
  {Brunner}}, \bibinfo {author} {\bibfnamefont {S.}~\bibnamefont {Pironio}},
  \bibinfo {author} {\bibfnamefont {A.}~\bibnamefont {Acin}}, \bibinfo {author}
  {\bibfnamefont {N.}~\bibnamefont {Gisin}}, \bibinfo {author} {\bibfnamefont
  {A.~A.}\ \bibnamefont {M\'ethot}},\ and\ \bibinfo {author} {\bibfnamefont
  {V.}~\bibnamefont {Scarani}},\ }\bibfield  {title} {\bibinfo {title} {Testing
  the dimension of {H}ilbert spaces},\ }\href
  {https://doi.org/10.1103/PhysRevLett.100.210503} {\bibfield  {journal}
  {\bibinfo  {journal} {Phys. Rev. Lett.}\ }\textbf {\bibinfo {volume} {100}},\
  \bibinfo {pages} {210503} (\bibinfo {year} {2008})}\BibitemShut {NoStop}%
\bibitem [{\citenamefont {Gallego}\ \emph {et~al.}(2010)\citenamefont
  {Gallego}, \citenamefont {Brunner}, \citenamefont {Hadley},\ and\
  \citenamefont {Ac\'{\i}n}}]{Gallego.etal2010}%
  \BibitemOpen
  \bibfield  {author} {\bibinfo {author} {\bibfnamefont {R.}~\bibnamefont
  {Gallego}}, \bibinfo {author} {\bibfnamefont {N.}~\bibnamefont {Brunner}},
  \bibinfo {author} {\bibfnamefont {C.}~\bibnamefont {Hadley}},\ and\ \bibinfo
  {author} {\bibfnamefont {A.}~\bibnamefont {Ac\'{\i}n}},\ }\bibfield  {title}
  {\bibinfo {title} {Device-independent tests of classical and quantum
  dimensions},\ }\href {https://doi.org/10.1103/PhysRevLett.105.230501}
  {\bibfield  {journal} {\bibinfo  {journal} {Phys. Rev. Lett.}\ }\textbf
  {\bibinfo {volume} {105}},\ \bibinfo {pages} {230501} (\bibinfo {year}
  {2010})}\BibitemShut {NoStop}%
\bibitem [{\citenamefont {Ahrens}\ \emph {et~al.}(2012)\citenamefont {Ahrens},
  \citenamefont {Badziag}, \citenamefont {Cabello},\ and\ \citenamefont
  {Bourennane}}]{Ahrens.etal2012}%
  \BibitemOpen
  \bibfield  {author} {\bibinfo {author} {\bibfnamefont {J.}~\bibnamefont
  {Ahrens}}, \bibinfo {author} {\bibfnamefont {P.}~\bibnamefont {Badziag}},
  \bibinfo {author} {\bibfnamefont {A.}~\bibnamefont {Cabello}},\ and\ \bibinfo
  {author} {\bibfnamefont {M.}~\bibnamefont {Bourennane}},\ }\bibfield  {title}
  {\bibinfo {title} {Experimental device-independent tests of classical and
  quantum dimensions},\ }\href {https://doi.org/10.1038/nphys2333} {\bibfield
  {journal} {\bibinfo  {journal} {Nat. Phys.}\ }\textbf {\bibinfo {volume}
  {8}},\ \bibinfo {pages} {592} (\bibinfo {year} {2012})}\BibitemShut {NoStop}%
\bibitem [{\citenamefont {Hendrych}\ \emph {et~al.}(2012)\citenamefont
  {Hendrych}, \citenamefont {Gallego}, \citenamefont {Mi{\v{c}}uda},
  \citenamefont {Brunner}, \citenamefont {Ac{\'\i}n},\ and\ \citenamefont
  {Torres}}]{Hendrych.etal2012}%
  \BibitemOpen
  \bibfield  {author} {\bibinfo {author} {\bibfnamefont {M.}~\bibnamefont
  {Hendrych}}, \bibinfo {author} {\bibfnamefont {R.}~\bibnamefont {Gallego}},
  \bibinfo {author} {\bibfnamefont {M.}~\bibnamefont {Mi{\v{c}}uda}}, \bibinfo
  {author} {\bibfnamefont {N.}~\bibnamefont {Brunner}}, \bibinfo {author}
  {\bibfnamefont {A.}~\bibnamefont {Ac{\'\i}n}},\ and\ \bibinfo {author}
  {\bibfnamefont {J.~P.}\ \bibnamefont {Torres}},\ }\bibfield  {title}
  {\bibinfo {title} {Experimental estimation of the dimension of classical and
  quantum systems},\ }\href {https://doi.org/10.1038/nphys2334} {\bibfield
  {journal} {\bibinfo  {journal} {Nat. Phys.}\ }\textbf {\bibinfo {volume}
  {8}},\ \bibinfo {pages} {588} (\bibinfo {year} {2012})}\BibitemShut {NoStop}%
\bibitem [{\citenamefont {Brunner}\ \emph {et~al.}(2013)\citenamefont
  {Brunner}, \citenamefont {Navascu\'es},\ and\ \citenamefont
  {V\'ertesi}}]{Brunner.etal2013}%
  \BibitemOpen
  \bibfield  {author} {\bibinfo {author} {\bibfnamefont {N.}~\bibnamefont
  {Brunner}}, \bibinfo {author} {\bibfnamefont {M.}~\bibnamefont
  {Navascu\'es}},\ and\ \bibinfo {author} {\bibfnamefont {T.}~\bibnamefont
  {V\'ertesi}},\ }\bibfield  {title} {\bibinfo {title} {Dimension witnesses and
  quantum state discrimination},\ }\href
  {https://doi.org/10.1103/PhysRevLett.110.150501} {\bibfield  {journal}
  {\bibinfo  {journal} {Phys. Rev. Lett.}\ }\textbf {\bibinfo {volume} {110}},\
  \bibinfo {pages} {150501} (\bibinfo {year} {2013})}\BibitemShut {NoStop}%
\bibitem [{\citenamefont {D'Ambrosio}\ \emph {et~al.}(2014)\citenamefont
  {D'Ambrosio}, \citenamefont {Bisesto}, \citenamefont {Sciarrino},
  \citenamefont {Barra}, \citenamefont {Lima},\ and\ \citenamefont
  {Cabello}}]{DAmbrosio.etal2014}%
  \BibitemOpen
  \bibfield  {author} {\bibinfo {author} {\bibfnamefont {V.}~\bibnamefont
  {D'Ambrosio}}, \bibinfo {author} {\bibfnamefont {F.}~\bibnamefont {Bisesto}},
  \bibinfo {author} {\bibfnamefont {F.}~\bibnamefont {Sciarrino}}, \bibinfo
  {author} {\bibfnamefont {J.~F.}\ \bibnamefont {Barra}}, \bibinfo {author}
  {\bibfnamefont {G.}~\bibnamefont {Lima}},\ and\ \bibinfo {author}
  {\bibfnamefont {A.}~\bibnamefont {Cabello}},\ }\bibfield  {title} {\bibinfo
  {title} {Device-independent certification of high-dimensional quantum
  systems},\ }\href {https://doi.org/10.1103/PhysRevLett.112.140503} {\bibfield
   {journal} {\bibinfo  {journal} {Phys. Rev. Lett.}\ }\textbf {\bibinfo
  {volume} {112}},\ \bibinfo {pages} {140503} (\bibinfo {year}
  {2014})}\BibitemShut {NoStop}%
\bibitem [{\citenamefont {Liu}\ \emph {et~al.}(2023)\citenamefont {Liu},
  \citenamefont {Meng}, \citenamefont {Xu}, \citenamefont {Zhou}, \citenamefont
  {Chen}, \citenamefont {Xu}, \citenamefont {Li}, \citenamefont {Guo},\ and\
  \citenamefont {Cabello}}]{Liu.etal2023}%
  \BibitemOpen
  \bibfield  {author} {\bibinfo {author} {\bibfnamefont {Z.-H.}\ \bibnamefont
  {Liu}}, \bibinfo {author} {\bibfnamefont {H.-X.}\ \bibnamefont {Meng}},
  \bibinfo {author} {\bibfnamefont {Z.-P.}\ \bibnamefont {Xu}}, \bibinfo
  {author} {\bibfnamefont {J.}~\bibnamefont {Zhou}}, \bibinfo {author}
  {\bibfnamefont {J.-L.}\ \bibnamefont {Chen}}, \bibinfo {author}
  {\bibfnamefont {J.-S.}\ \bibnamefont {Xu}}, \bibinfo {author} {\bibfnamefont
  {C.-F.}\ \bibnamefont {Li}}, \bibinfo {author} {\bibfnamefont {G.-C.}\
  \bibnamefont {Guo}},\ and\ \bibinfo {author} {\bibfnamefont {A.}~\bibnamefont
  {Cabello}},\ }\bibfield  {title} {\bibinfo {title} {Experimental test of
  high-dimensional quantum contextuality based on contextuality
  concentration},\ }\href {https://doi.org/10.1103/PhysRevLett.130.240202}
  {\bibfield  {journal} {\bibinfo  {journal} {Phys. Rev. Lett.}\ }\textbf
  {\bibinfo {volume} {130}},\ \bibinfo {pages} {240202} (\bibinfo {year}
  {2023})}\BibitemShut {NoStop}%
\bibitem [{\citenamefont {Spekkens}\ \emph {et~al.}(2009)\citenamefont
  {Spekkens}, \citenamefont {Buzacott}, \citenamefont {Keehn}, \citenamefont
  {Toner},\ and\ \citenamefont {Pryde}}]{Spekkens.etal2009}%
  \BibitemOpen
  \bibfield  {author} {\bibinfo {author} {\bibfnamefont {R.~W.}\ \bibnamefont
  {Spekkens}}, \bibinfo {author} {\bibfnamefont {D.~H.}\ \bibnamefont
  {Buzacott}}, \bibinfo {author} {\bibfnamefont {A.~J.}\ \bibnamefont {Keehn}},
  \bibinfo {author} {\bibfnamefont {B.}~\bibnamefont {Toner}},\ and\ \bibinfo
  {author} {\bibfnamefont {G.~J.}\ \bibnamefont {Pryde}},\ }\bibfield  {title}
  {\bibinfo {title} {Preparation contextuality powers parity-oblivious
  multiplexing},\ }\href {https://doi.org/10.1103/PhysRevLett.102.010401}
  {\bibfield  {journal} {\bibinfo  {journal} {Phys. Rev. Lett.}\ }\textbf
  {\bibinfo {volume} {102}},\ \bibinfo {pages} {010401} (\bibinfo {year}
  {2009})}\BibitemShut {NoStop}%
\bibitem [{\citenamefont {Cubitt}\ \emph {et~al.}(2010)\citenamefont {Cubitt},
  \citenamefont {Leung}, \citenamefont {Matthews},\ and\ \citenamefont
  {Winter}}]{Cubitt.etal2010}%
  \BibitemOpen
  \bibfield  {author} {\bibinfo {author} {\bibfnamefont {T.~S.}\ \bibnamefont
  {Cubitt}}, \bibinfo {author} {\bibfnamefont {D.}~\bibnamefont {Leung}},
  \bibinfo {author} {\bibfnamefont {W.}~\bibnamefont {Matthews}},\ and\
  \bibinfo {author} {\bibfnamefont {A.}~\bibnamefont {Winter}},\ }\bibfield
  {title} {\bibinfo {title} {Improving zero-error classical communication with
  entanglement},\ }\href {https://doi.org/10.1103/PhysRevLett.104.230503}
  {\bibfield  {journal} {\bibinfo  {journal} {Phys. Rev. Lett.}\ }\textbf
  {\bibinfo {volume} {104}},\ \bibinfo {pages} {230503} (\bibinfo {year}
  {2010})}\BibitemShut {NoStop}%
\bibitem [{\citenamefont {Xu}\ \emph {et~al.}(2021)\citenamefont {Xu},
  \citenamefont {Yu},\ and\ \citenamefont {Kleinmann}}]{Xu.etal2021}%
  \BibitemOpen
  \bibfield  {author} {\bibinfo {author} {\bibfnamefont {Z.-P.}\ \bibnamefont
  {Xu}}, \bibinfo {author} {\bibfnamefont {X.-D.}\ \bibnamefont {Yu}},\ and\
  \bibinfo {author} {\bibfnamefont {M.}~\bibnamefont {Kleinmann}},\ }\bibfield
  {title} {\bibinfo {title} {State-independent quantum contextuality with
  projectors of nonunit rank},\ }\href
  {https://doi.org/10.1088/1367-2630/abe6e3} {\bibfield  {journal} {\bibinfo
  {journal} {New J. Phys.}\ }\textbf {\bibinfo {volume} {23}},\ \bibinfo
  {pages} {043025} (\bibinfo {year} {2021})}\BibitemShut {NoStop}%
\bibitem [{\citenamefont {Dattorro}(2005)}]{Dattorro2005}%
  \BibitemOpen
  \bibfield  {author} {\bibinfo {author} {\bibfnamefont {J.}~\bibnamefont
  {Dattorro}},\ }\href@noop {} {\emph {\bibinfo {title} {Convex Optimization
  and Euclidean Distance Geometry}}}\ (\bibinfo  {publisher} {Meboo Publishing
  USA},\ \bibinfo {year} {2005})\BibitemShut {NoStop}%
\bibitem [{\citenamefont {Boyd}\ and\ \citenamefont
  {Vandenberghe}(2004)}]{Boyd.Vandenberghe2004}%
  \BibitemOpen
  \bibfield  {author} {\bibinfo {author} {\bibfnamefont {S.}~\bibnamefont
  {Boyd}}\ and\ \bibinfo {author} {\bibfnamefont {L.}~\bibnamefont
  {Vandenberghe}},\ }\href@noop {} {\emph {\bibinfo {title} {Convex
  Optimization}}}\ (\bibinfo  {publisher} {Cambridge University Press, New
  York},\ \bibinfo {year} {2004})\BibitemShut {NoStop}%
\bibitem [{\citenamefont {Horn}\ and\ \citenamefont
  {Johnson}(2012)}]{Horn.Johnson2012}%
  \BibitemOpen
  \bibfield  {author} {\bibinfo {author} {\bibfnamefont {R.~A.}\ \bibnamefont
  {Horn}}\ and\ \bibinfo {author} {\bibfnamefont {C.~R.}\ \bibnamefont
  {Johnson}},\ }\href@noop {} {\emph {\bibinfo {title} {Matrix Analysis}}}\
  (\bibinfo  {publisher} {Cambridge University Press, Cambridge},\ \bibinfo
  {year} {2012})\BibitemShut {NoStop}%
\end{thebibliography}%

\end{document}